%% file: main.tex
\documentclass{article}
\usepackage[utf8]{inputenc}

\usepackage{fullpage}
\usepackage[english]{babel}
\usepackage{graphicx}
\usepackage{amsmath,amsthm,bbm,amssymb,multirow,times,bm}
\usepackage{url}
\usepackage{hyperref}
\hypersetup{
    colorlinks=true,
    linkcolor=blue,
    filecolor=magenta,      
    urlcolor=cyan,
}

\usepackage{algorithm,algpseudocode}
\usepackage{color,xcolor}
\usepackage{csquotes}

\usepackage[normalem]{ulem}
\usepackage{graphicx}
\usepackage{subcaption}

\makeatletter
\newcounter{parentsubcaption}

\makeatother

\usepackage[backend=biber,style=alphabetic]{biblatex}

\algrenewcommand\algorithmicrequire{\textbf{Input:}}
\algrenewcommand\algorithmicensure{\textbf{Output:}}

\newcommand{\eps}{\varepsilon}
\newcommand{\littleo}[1]{o\left( #1 \right)}

\newcommand{\nrm}[1]{\left\lVert #1 \right\rVert}

\newcommand{\bra}[1]{ \left\langle{#1}\right|}
\newcommand{\ket}[1]{\left|{#1}\right\rangle}
\newcommand{\braket}[1]{\left\langle{#1}\right\rangle}

\definecolor{blue}{rgb}{0,0.2,1}

\definecolor{red}{rgb}{0.9,0,0}

\newcommand{\Ord}[1]{\mathcal{O}\left( #1 \right)}
\newcommand{\tOrd}[1]{\widetilde{ \mathcal{O}}\left( #1 \right)}

\newcommand{\norm}[1]{\left\lVert#1\right\rVert}

\definecolor{Pr}{rgb}{0.4,0.3,0.9}
\definecolor{Ale}{rgb}{0.1,0.5,0.5}

\usepackage{pdfcomment}

\theoremstyle{plain}
\newtheorem{oracle}{Oracle}

\newtheorem{lemma}{Lemma}
\newtheorem*{lemma*}{Lemma}
\newtheorem*{result*}{Results}
\newtheorem{defn}{Definition}
\newtheorem{theorem}{Theorem}

\newtheorem{assume}{Assumption}

\newtheorem{remark}{Remark}

\def\be{\begin{eqnarray}}
\def\ee{\end{eqnarray}}

\newcommand{\F}{\mathbb{F}}

\usepackage{authblk}

\title{Quantum computational finance: martingale asset pricing in a linear programming setting}
\title{Quantum computational finance: martingale asset pricing for incomplete markets}

\author[1]{Patrick Rebentrost \thanks{cqtfpr@nus.edu.sg}}
\author[1,4]{Alessandro Luongo
\thanks{ale@nus.edu.sg}}
\author[2]{Samuel Bosch}
\author[2,3]{Seth Lloyd}
\affil[1]{Centre for Quantum Technologies, National University of Singapore}
\affil[2]{Research Laboratory of Electronics, Massachusetts Institute of Technology} 
\affil[3]{Department of Mechanical Engineering, Massachusetts Institute of Technology }
\affil[4]{Inveriant, Singapore}

\begin{document}

\maketitle

\begin{abstract}
A derivative is a financial security whose value is a function of underlying traded assets and market outcomes. 
Pricing a financial derivative involves setting up a market model, finding a martingale 
(``fair game") probability measure for the model from the given asset prices, and using that probability measure to price the derivative.
When the number of underlying assets and/or the number of market outcomes in the model is large, 
pricing can be computationally demanding.
We show that a variety of quantum techniques can be applied to the pricing problem in finance, with a particular focus on incomplete markets.
We discuss three different methods that are distinct from previous works: they do not use the quantum algorithms for Monte Carlo estimation and they extract the martingale measure from market variables akin to bootstrapping, a common practice among financial institutions.
The first two methods are based on a formulation of the pricing problem into a linear program and are using respectively the quantum zero-sum game algorithm and the quantum simplex algorithm as subroutines. For the last algorithm, we formalize a new market assumption milder than market completeness for which quantum linear systems solvers can be applied with the associated potential for large speedups. As a prototype use case, we conduct numerical experiments in the framework of the Black-Scholes-Merton model.
\end{abstract}

\section{Introduction}

Consider the following simplified version of the ``parable of the bookmaker" \cite{baxter_rennie_1996}.  A bookmaker is taking bets on a horse race between two horses A and B. Assume that, without quoting any odds, the bookmaker has received $\$5000$ bets on a win of horse A and $\$10000$ bets on a win of horse B. The bookmaker's desire is to have a sustainable business based on fees and not to use the money for gambling. 
What are the correct odds for the bookmaker to quote to the betters? The market has decided that horse A is half as likely to win over horse B, i.e., the implied probability of winning is $1/3$ vs. $2/3$. Hence the bookmaker quotes the odds of $2-1$, which means that if horse A wins each $\$1$ bet on horse A will turn into $\$2$, while bets on horse B are lost. If horse B wins, each $\$1$ bet on horse B will turn into $\$1.5$, while bets on horse A are lost. Given these odds for the bets taken, the bookmaker has a payoff of $\$0$ for each scenario, which is called \textit{risk neutral}. Note that the result of zero payoffs for all scenarios is independent of the true probabilities of winning for A or B.  

Financial markets show some similarities with the betting on horse races, and often involve highly complex scenarios and challenging computational problems.
A significant part of computational finance involves the pricing of financial 
securities such as derivatives \cite{Glasserman2003,Follmer2004,Hull2012}. 
Some of these securities (often called over-the-counter, OTC, derivatives) are  tailor-made between two parties, and thus the price of these instruments is not readily quoted on a liquid primary market \cite{cerny2009mathematical}. 
Such securities have complex payoffs that usually depend on underlying liquid assets, each introducing its own dynamics due to market and economic fluctuations. In that sense, these securities can be called exotic and illiquid. 
Exotic securities are usually not frequently traded to allow adequate price determination by the market itself.
Pricing can become computationally expensive, especially when the number of underlying liquid assets is large and/or the pricing model takes into account a large
number of possible market outcomes.

The Black-Scholes-Merton framework \cite{Black1973,Merton1973,Hull2012} allows for a concise formula for the price of a call option on a single stock and other options. The market model in the simplest formulation is  
a single risky stock driven by a random process (e.g., a Brownian motion) and one safe asset (``bond" or bank account).
One can consider a generalized situation for markets with a large amount of underlying assets and a large sample space on which arbitrary driving processes are defined.
The number of assets $N$ can be large, say all the stocks traded on the stock market and in addition simple frequently traded derivatives on these stocks.
The size $K$ of the underlying sample space can also be large. Such a sample space could be given by the discretization of a continuous model space for a multidimensional random walk used in finance, or by the exact outcomes of the NCAA basketball
playoffs (which have $\sim 2^{68}$ possibilities for $68$ teams)
in a betting scenario.   
Large outcome spaces and the presence of many traded securities can make pricing hard in many cases and motivate the use of quantum computers. The  \emph{risk-neutral pricing} framework is employed in practice in stock, interest rate, and credit derivative markets.
Given a liquid financial market trading a number of assets, the task is to price a new financial instrument which is not traded in the market. To price this instrument, so-called risk-neutral probability measures are extracted from the market variables. Such measures turn the discounted price process of each of the assets into a \emph{martingale} and are used to compute the price of the derivative via an expectation value.

\subsection{Main results} 
In this work, we discuss algorithms in quantum computational finance which can be used for the pricing of certain illiquid financial securities. We discuss quantum algorithms that can be used in incomplete markets with the potential for significant speedups over classical algorithms. The algorithms presented here are the first to extract the martingale measures from market data. This is akin to bootstrapping: the process of extracting information from market variables, a common practice among financial institutions. We describe in Section~\ref{sec:introtomartingale} the concepts of Arrow securities, price system, martingale measures, including the existence of martingale measures from a no-arbitrage assumption, and complete markets. 
In Section~\ref{sec:riskneutralpricing} we introduce a pricing framework where quantum linear programming algorithms can be employed. In this framework, the Arrow securities provide the basis for constructing a probability measure, called martingale measure, which encapsulates the market information and allows to find the price of new securities. 
The output of financially relevant results
is usually a range of possible prices of the new asset.

In Section~\ref{sec:zerosumgames}, we adapt to our settings the quantum algorithm for zero-sum games \cite{Apeldoorn2019} which is asymptotically one of the fastest algorithms for solving linear programs in terms of query complexity. We give the conditions for which the quantum algorithm performs better than the classical one. We further explore in Section~\ref{secExampleBlackScholes} the performances of the zero-sum game algorithm on market instances derived from Black-Scholes-Merton (BSM) models. Working in this model eases the space requirements of the input to be only $O(N)$, compared to $O(NK)$ in the general case. 
As an alternative to the zero-sum game algorithm, in Appendix~\ref{sec:simplexpricing} we discuss the usage of a quantum version of the simplex method for pricing. Simplex methods are often used for solving linear programs: they provide exponential worst-case complexity, but they can be very efficient in practice.
We adapt a previous quantum algorithm for the simplex method and compare to the quantum zero-sum game algorithm. 

In Section~\ref{appendix:potentialQLA} we explore the use of quantum matrix inversion algorithms for pricing and extracting the martingale measure. We introduce the definition of a \emph{least-squares market}, and prove that market completeness is a sufficient condition for having a least-square market (Theorem~\ref{thm:spseudoinvisq}). 
As the converse is usually not true, a least-squares market is generally a weaker assumption than a complete market. We show that with appropriate assumptions on the data input, quantum linear systems solvers are able to provide the derivative price for least-squares markets with the potential of a large speedup over classical algorithms.

All the quantum algorithms assume quantum access to a matrix called payoff matrix $S\in \mathbb{R}_+^{N+1 \times K}$ (which we define formally in the main text) whose rows represent the assets traded in the market that we consider, and whose columns represent the possible events. If an event indexed by $j \in [K]$ happens, the asset $i \in [N+1]$ takes on the value $S_{ij}$. The current value of the assets are described in the vector $\bm \Pi \in \mathbb{R}_+^{N+1}$. The performance of the different algorithms used in this work can be summarized informally as follows, (here  $\kappa(S)$ is the condition number of the payoff matrix).

\begin{result*}[Informal - Theorems~\ref{thm:pricingzsg}, \ref{thm:pricingsimplex}, \ref{thm:pricingqla}]
Assume to have appropriate quantum access to a price vector $\bm \Pi \in \mathbb{R}^{N+1}$ of assets, the payoff matrix $S \in \mathbb{R}^{N+1 \times K}$ of a single-period market, and a derivative $D$. Then, with a relative error $\epsilon$
\begin{enumerate}
    \item we can estimate the range of implied present values of the derivative 
    \begin{itemize}
        \item with a quantum zero-sum games algorithm, with cost $O\left(\sqrt{N} + \sqrt{K})/\epsilon^3\right)$,
        \item with a quantum simplex algorithm with cost $O\left(\kappa(S_B)\sqrt{NK}+N\right)$ per iteration, a number of iterations proportional to the classical algorithm, and a final step of $O\left(\frac{\kappa(S)\|S\|_F}{\epsilon}\right)$ operations, 
    \end{itemize}
    \item if the market is a least-squares market (Definition \ref{def:assumeL2}) we can estimate the price of the derivative with a quantum matrix-inversion algorithm with cost $O\left(\frac{\kappa(S)\|S\|_F}{\epsilon}\right)$ .
\end{enumerate}
\end{result*}

\paragraph{Related works} 
Quantum speedups have been discussed for a large set of problems, 
including integer factoring, hidden subgroup problems, unstructured search \cite{Nielsen2000}, solving linear systems \cite{harrow2009quantum} and for convex optimization \cite{Brandao2017,Apeldoorn2017}. Financial applications are receiving increased attention. A quantum algorithm for Monte-Carlo asset pricing given knowledge of the martingale measure was shown in \cite{Rebentrost2018finance}. Similarly, an algorithm for risk management \cite{Woerner2018} and further works on derivative pricing \cite{Stamatopoulos2019} with amplitude estimation and principal component analysis \cite{Martin2019} were shown with proof-of-principle implementations on quantum hardware. More in-depth reviews are given in \cite{Orus2019,bouland2020prospects,herman2022survey}. 
Zero-sum games are a particular case of minmax games, specifically they are the $\ell_1-\ell_1$ minmax games. More general quantum algorithms for solving minmax games can be found in 
\cite[Theorem 7]{li2019sublinear} \cite{li2021sublinear}.

\subsection{Preliminaries}
Define $[N] := \{ 1,\dots, N \} $. We use boldface for denoting a vector $\bm v$,  where the $i$-th entry of the vector is denoted by $v_i$. With $\bm v_i$ we denote the $i$-th vector in a list of vectors. 
Denote by $\mathbbm R_+$ the non-negative real numbers.
The Kronecker delta is $ \delta_{\omega \omega'}$ for some integers $\omega, \omega'$.  Define the simplex of non-negative, $\ell_1$-normalized $K$-dimensional vectors and its interior by 
\be
\Delta^{K} := \left \{ \bm p \in \mathbbm R^K : \sum_{k=1}^K p_k =1, p_k \geq 0\right \},\quad {\rm int}(\Delta^{K}) := \left \{ \bm p \in \mathbbm R^K : \sum_{k=1}^K p_k =1, p_k > 0\right \}.
\ee
For a matrix $M$, the spectral norm is $\Vert M\Vert$, the Frobenius norm is $\Vert M \Vert_F$, and the condition number is  $\kappa(M):=\Vert M\Vert \Vert M^+\Vert$, where the pseudoinverse is denoted by $M^+$.  
Denote by $M_{i,\cdot}$ the $i$-th row of the matrix. We use the notation $\tOrd{a}$ to hide factors that are polylogarithmic in $a$ and the notation $\widetilde{\mathcal{O}}_b(a)$ to hide factors that are polylogarithmic in $a$ and $b$. 
We assume basic knowledge of quantum computing and convex optimization (in particular linear programming), and we use the same notation that can be found in \cite{Nielsen2000} and \cite{Boyd2004}.

We refer to Appendix \ref{app:probability} for a compressed introduction to aspects of probability theory. In general, a probability space is given by a triple $(\Omega, \Sigma, \mathbb{P})$, see Definition \ref{defn:probability-space} of Appendix \ref{app:probability}. 
Our first assumption is the existence of a sample space $\Omega$ which is finite (and hence countable). We take the $\sigma$-algebra $\Sigma$ to be the set of all subsets of $\Omega$. 
For most of the paper, we do not require the knowledge of the probability measure $\mathbb{P}$, a setting which is called Knightian uncertainty \cite{knight1921risk}. We only require the assumption that $\mathbb{P}[\{\omega \}] > 0$ for all $\omega \in \Omega$: it is reasonable to exclude singleton events from $\Omega$ that have no probability of occurring. 
Equivalence of two probability measures $\mathbbm P$ and $\mathbbm Q$ is  
if, for $A \in \Sigma$, $\mathbbm P[A]=0$ if and only if $\mathbbm Q[A] = 0$. In our simplified context, this amounts to the statement that if, for $\omega \in \Omega$, $\mathbbm P[\{ \omega \}]=0$ if and only if $\mathbbm Q[\{ \omega \}] = 0$. Since we assume that $\mathbbm P[\{ \omega \}]>0$, equivalent measures $\mathbbm Q$ are all measures for which $\mathbbm Q[\{ \omega \}] > 0$. 
We denote with $\mathbb{E}_\mathbb{P}[X]$ the expectation value of the random variable $X$ with respect to the probability measure $\mathbb{P}$. We refer to Appendix \ref{app:probability} for formal definitions in probability theory used in this work. In Appendix \ref{app:linprogoptimization} we recall some useful definition in linear programming and optimization, and in Appendix \ref{subsec:quantumsubroutines} we report some useful definitions and subroutines in quantum numerical linear algebra.

\section{Market models, martingale asset pricing, and fundamental theorems of finance}\label{sec:introtomartingale}

Consider Appendix \ref{appMarket} for some philosophical aspects of modeling financial markets. 
In this work, we consider a single-period model, which consists of a present time and a future time.
The random variables of the future asset prices are defined on a sample space $\Omega$.
We assume that  $\Omega$ 
is finite with $K:=\vert \Omega \vert$ elements. 
We consider a market with $N+1$ primary assets, which are $N$ risky assets and one safe asset.
The risky assets could be $N_{\rm stocks}$ stocks and possibly other financial instruments on these stocks, such as simple call options, e.g., $\max(x-Z,0)$, or options
involving several stocks, e.g., $\max(x+y-Z,0)$, where $Z$ is the strike price. 
The number of such possible options in principle scales with the number of subsets of the set of stocks $\Ord{2^{N_{\rm stocks}}}$, and hence can be quite large. 
The safe asset acts like a bank account and models the fact that money at different time instances has a different present value due to inflation and the potential for achieving some return on investment. Discount factors are used for comparing monetary values at different time instances.  To simplify the notation, we assume in this work that the 
discount factors are $1$, which corresponds to a risk-free interest rate of $0$ \cite{Follmer2004}. Hence, our bank account does not pay any interest. 
At the present time, the prices of all assets are known. 
The prices are denoted by $\Pi_i \in \mathbbm R_+$ for $i\in [N+1]$, where $\Pi_1=1$ for the safe asset.
We assume that we can buy and short-sell an arbitrary number and fractions of them at no transaction cost. The prices at the future time are given by the random variables $S_i : \Omega \to \mathbb R_+$, for $i \in [N+1]$, where $S_1 : \Omega \to 1$ is not random for the safe asset. 

In our setting, risk-neutral pricing will have two conceptual steps.
The first step finds probability measures under which the existing asset prices are martingales.
Such probability measures are derived from the market prices of the existing assets, and are
denoted by $\mathbbm Q$ in the probability theory context, $\bm q$ in the vector notation, and are elements of a set denoted by  $\mathcal Q$.
Measures $\mathbbm Q \in \mathcal Q$ should be \emph{equivalent} to the original probability measure $\mathbb{P}$, however, the algorithms here will achieve equivalency only approximately.
As a second step, probability measures $\mathbb{Q}\in \mathcal Q$ are used to determine the range of fair prices of the non-traded financial security by computing expectation values.
In the remainder of this section, we discuss some background in financial asset pricing. 
In particular, we recall the concept of Arrow security, the martingale property of financial assets, the definition of arbitrage, and the existence of martingale probability measures. 

\paragraph{Arrow security} The notion of an Arrow security \cite{lengwiler2004microfoundations,varian1987arbitrage} is useful in the pricing setting. 
Arrow securities are imaginary securities
that pay one unit of currency, say $\$1$, in case some event occurs and pay nothing otherwise. Albeit they are not traded in most of the markets, they have historical and theoretical relevance in finance and asset pricing. They allow us to express the future stock prices as a sum of Arrow securities multiplied by respective scaling factors. This separation then naturally leads to a matrix description of the problem of pricing a derivative. 
 
\begin{defn}[Arrow security \cite{varian1987arbitrage}] 
The Arrow security for $\omega\in \Omega$ is the random variable $A_{\omega} : \Omega \mapsto \{0,1\}$ such that in the event $\omega' \in \Omega$ the payoff of the security is: $A_{\omega} (\omega') = \delta_{\omega \omega'}.$
\end{defn}
Since Arrow securities are often not traded in the market, they do not have a price. 
Their prices may be inferred from the traded securities if these securities can be decomposed into Arrow securities. 
As we will see, pricing the Arrow securities amounts to finding a probability measure, which in turn can be used for pricing other assets. 
Arrow securities can be considered as a basis of vectors from which the other financial assets are expanded, as in the next definition. 
For this definition, we also require the concept of redundancy: here, a redundant asset is a primary asset $i \in [N+1]$ which can be expressed as a linear combination of the other $N$ primary assets. 
\begin{defn}[Payoffs, payoff matrix, price system] \label{def:payoff-matrix}
Let there be given one safe asset with current price of $\Pi_1 = 1$ and payoff $S_1 : \Omega \mapsto 1$.
Let there be given $N$ non-redundant assets with current price $\Pi_i \in \mathbbm R_+$ and future payoff $S_i : \Omega \mapsto \mathbbm R_+$ for $i\in \{ 2,\cdots, N+1\}$. 
We denote by $\bm{S}: \Omega \to \mathbbm R_+^{N+1}$ the associated vector of asset payoffs and by $\bm \Pi = [1, \Pi_2, \cdots, \Pi_{N+1}]^T$ the price vector.
Define the payoff matrix $S \in \mathbbm R_+^{N+1\times K}$ such that the assets are expanded as
\begin{equation} \label {eqArrowExpansion}
 S_{1 \omega} = 1\quad \text{for}\ \omega \in \Omega, \quad S_i = \sum_{\omega \in \Omega} S_{i \omega} A_{\omega}\quad \text{for}\ i\in \{ 2,\cdots, N+1\}.
\end{equation}
Finally, we call $(\bm \Pi, S)$ a price system. 
\end{defn}
The matrix $S$ specifies the asset price for each asset for each market outcome.
An interpretation of a single row of the payoff matrix is as a portfolio that holds the amounts $S_{i \omega}$ in Arrow security $A_\omega$. Hence, the random variable $S_i$ is \textit{replicated} with a portfolio of Arrow securities. However, because of the non-redundancy assumption, it cannot be replicated with the other primary assets. Non-redundancy implies that the matrix $S$ has always rank $N+1$. 

\paragraph{Martingale measure} Martingales are defined in Definition \ref{def:martingale}. For the asset prices, the martingale property means that the current asset prices are fair, that is, they reflect the discounted expected future payoffs. The martingale property is a property of a random variable relative to a probability measure. Martingale probability measures are often denoted by $\mathbbm Q$ in the financial literature. 
We want to find such martingale probability measures which are equivalent to the original probability measure $\mathbbm P$. 
In our setting, equivalency holds if for the new probability measure all probabilities are strictly positive. 
The probability measure is hence $\in {\rm int}(\Delta^K)$ and denoted by $\bm q \in {\rm int}(\Delta^K)$. The martingale property of Definition \ref{def:martingale} here reads as:
\begin{equation} \label{eqMartingale}
\bm \Pi = \mathbbm{E}_{\bm q}[\bm S]. 
\end{equation}
The corresponding probability measure $\bm q$ 
is called equivalent martingale measure 
or risk-neutral measure, whose existence is discussed below. 
Finding a risk-neutral measure is equivalent to finding the prices of the Arrow securities. The expectation value of an Arrow security is the probability of the respective event, 
\begin{equation} \label{eqArrowProperty}
\mathbbm{E}_{\bm q}[A_{\omega}]=  q_{\omega}.
\end{equation}
Putting together the martingale property for the asset prices of Eq.~(\ref{eqMartingale}),
the expansion into Arrow securities of Eq.~(\ref{eqArrowExpansion}), and 
the property of the Arrow securities of Eq.~(\ref{eqArrowProperty})
we obtain for the prices of the traded securities:
\begin{equation}
\Pi_i= \mathbbm{E}_{\bm q}[S_i] = \sum_\omega S_{i\omega}  q_\omega.
\end{equation}
In matrix form this equation reads as $\bm \Pi = S\bm q$. 
The set of solutions that are contained in the simplex
is given by
\be
{\mathcal Q} :=  \{ \bm q \in {\rm int}(\Delta^K) : S \bm q = \bm \Pi \}.
\ee
Depending on the problem this set is either empty, contains single or multiple elements, or an infinite number of elements.
The existence of a solution is discussed below via a no-arbitrage argument. 
If measures $\bm q \in \mathcal Q$ exist, they can be obtained with feasibility methods \cite{Boyd2004}.
We note that linear programming formulations do not admit optimization over open sets, i.e., cannot work with strict inequality constraints. 
Thus we 
have to work with the closed convex set
\be
\overline{\mathcal Q} :=  \{ \bm q \in \Delta^K : S \bm q = \bm \Pi \}.
\ee
As all algorithms in this work are approximation algorithms, the sets  ${\mathcal Q}$ and $\overline{\mathcal Q}$ will be considered the same, up to the error tolerance for all the components.

\paragraph{Existence of martingale measure} Now we discuss the existence of a martingale measure, a property that is closely related to the non-existence of arbitrage in the market. Arbitrage is the existence of portfolios that allow for gains without the corresponding risk.  A standard assumption in finance is that the market model under consideration is arbitrage-free. 
No arbitrage 
implies the existence of a non-negative solution $\bm q> 0$ via Farkas' lemma.
\begin{lemma}[Farkas' lemma \cite{Boyd2004}] \label{lemFarkas}
For a matrix $S \in \mathbbm{R}^{(N+1)\times K}$ and a vector 
$\bm \Pi \in \mathbbm{R}^{N+1}$ exactly one of the two 
statements hold: 
\begin{itemize} 
  \item there exists a
 $\bm q \in \mathbbm{R}^{K}$ such that $S\bm q = \bm \Pi$ and $\bm q > \bm 0$, 
\item 
there exists a $\bm v \in \mathbbm{R}^{N+1}$ such that $\bm v^T \bm \Pi < 0$ and $S^T \bm v > \bm 0 $. 
\end{itemize}
\end{lemma}

Note that we are using a slightly modified version of Farkas' Lemma (whose proof can be found in Appendix~\ref{app:probability}), which usually considers only that $\bm q \geq 0$ and $S^T \bm v \geq 0$. Expressing Lemma \ref{lemFarkas} in words, in the first case the vector $\bm \Pi$ is spanned by a positive combination of the columns of $S$.
In the second case,  there exists a separating hyperplane defined by a normal vector 
$\bm v$ that separates $\bm \Pi$ from the convex cone spanned by the columns of $S$.
In the present financial framework, we can think of $\bm v$ as a portfolio that is allocated in the different assets. Then, this lemma is immediately applicable to an arbitrage discussion.
\begin{defn}[Arbitrage portfolios and no-arbitrage assumption] \label{def:arbitrage}
An arbitrage portfolio
is defined as a portfolio with non-positive present value $\bm v^T \bm \Pi \leq 0$ and non-negative future value for all market outcomes $S^T \bm v\geq 0$, and at least one $j\in [K]$ such that $(S^T\bm v)_j > 0$. The no-arbitrage assumption, denoted by ({\bf NA}), states that in a market model such portfolios do not exist. 
\end{defn}
By this definition, $({\bf NA})$ expressed as a mathematical statement is
\be
\nexists \: \bm v \in \mathbbm{R}^{N+1} \text{ such that } \bm v^T \bm \Pi \leq 0, \: S^T \bm v \geq \bm 0, \text{ and } \exists j\in [K] \text{ such that } (S^T\bm v)_j > 0.
\ee
It is easy to see that
\begin{equation} \label{eqNAImplies}
({\bf NA}) \Rightarrow \nexists \:  \bm v \in \mathbbm{R}^{N+1} \text{ such that } \bm v^T \bm \Pi < 0 \text{ and } S^T \bm v > \bm 0,
\end{equation} 
since the right hand side is a special case of $({\bf NA})$.
Then, if $({\bf NA})$ holds Farkas' lemma implies that case (i) holds, as case (ii) is ruled out by the statement in  Eq. (\ref{eqNAImplies}). Since there exists a $\bm q > 0$ such $S\bm q = \bm \Pi$, the set $\mathcal Q$ is by its definition not empty.  In conclusion, it holds that
\begin{equation}
({\bf NA}) \Rightarrow \mathcal |Q| \not = 0.
\end{equation}
Together with the reverse direction, this result is a formulation of the ``first fundamental theorem of asset pricing". There are special cases for which $\mathcal |Q| =1$, called complete markets. 
Before we define complete markets, we formalize the concept of financial derivatives.

\paragraph{Financial derivatives}
A financial derivative in the context of this work is a new asset whose payoff can be expanded in terms of Arrow securities and which does not have a current price.
\begin{defn}[Financial derivative~\cite{Follmer2004}]\label{def:derivative}
A derivative is a financial security given by a random variable $D : \Omega \to \mathbb{R}_+$. 
The derivative $D$ can be written as an expansion of Arrow's securities with expansion coefficients $D_\omega, \forall \omega \in \Omega$, as:
\begin{equation}\label{eq:derivative as arrow}
D=\sum_{\omega  \in \Omega} D_{\omega} A_\omega. 
\end{equation}
We denote by $\bm D \in \mathbb{R}_+^K$ the vector of $D_\omega$ arising from the random variable $D$.
\end{defn}
This definition includes the usual case when the derivative has a payoff that is some function of one or multiple of the existing assets. 
As an example, let the function be $f: \mathbb R_+ \to \mathbb R_+$ and the derivative be $D_{[f]} := f(S_j)$, for some asset $j \in [N+1]$. Then, we have the expansion in Arrow securities $D_{[f]} =\sum_{\omega  \in \Omega} f(S_{j \omega}) A_\omega$.

\paragraph{Complete markets}

We define another common assumption for financial theory. Market completeness means that every possible payoff is replicable with a portfolio of the existing assets.  Hence it is an extremal case where every derivative is redundant.

\begin{defn} [Complete market {\cite[Definition 1.39]{Follmer2004}}]\label{assumeComplete} 
An arbitrage-free market model is called complete if every derivative $D$ (Definition~\ref{def:derivative}) is replicable, i.e., if there exists a vector $\bm \xi \in \mathbbm R^{N+1}$ such that $\bm D = S^T \bm{\xi}$ .  
\end{defn}
A fundamental result relates complete markets to the set of equivalent martingale measures $\mathcal Q$. For complete markets, the martingale measure is unique. 
\begin{theorem}[Second fundamental theorem of asset pricing {\cite[Theorem 1.40]{Follmer2004}}]\label{thm:sftap}
An arbitrage-free market model is complete if and only if
$\vert \mathcal{Q} \vert = 1$.
\end{theorem}
Thus, there exists a unique martingale measure if and only if every financial derivative is redundant. 
In an incomplete market, derivatives usually cannot be replicated.

\section{Risk-neutral pricing of financial derivatives as a linear program}\label{sec:riskneutralpricing}\label{sectionLP}

For pricing a derivative, we use any risk-neutral probability vector $\bm q \in \mathcal Q$, if the set is not empty.
Lacking any other information or criteria any element in $\mathcal Q$ is a valid pricing measure. 
The derivative is priced by computing the expectation value of the future payoff under the risk-neutral measure $\bm q$, (discount factors are taken to be $1$)
\begin{equation}
\Pi_{D,\bm q} := \mathbbm{E}_{\bm q}[{D}] = \sum_{\omega\in\Omega} 
D_\omega q_\omega.
\end{equation}
The price $\Pi_{D,\bm q}$ is an inner product between the pricing probability measure $\bm q$ and the vector describing the random variable of the financial derivative. 
Since $\mathcal Q$ is convex \cite{Follmer2004}, it follows that the set of possible prices of the derivative is a convex set, i.e., an interval. We can formulate as a convex optimization problem to evaluate the maximum and minimum price $\Pi_{{D},\min}$ and $\Pi_{{D},\max}$, respectively, which are possible over all the martingale measures. We obtain two optimization programs:
\be
 \Pi_{{D},\max} &:=& \max_{\bm q \in \mathcal Q} \mathbbm{E}_{\bm q}[{D}], \label{eqLPmax} \\ 
 \Pi_{{D},\min}  &:=& \min_{\bm q \in \mathcal Q} \mathbbm{E}_{\bm q}[{D}] \label{eqLPmin}.
\ee
When the solution of the optimization problems is a broad range of prices,  additional constraints can be added to narrow the range of admissible prices. We discuss a simple case study in Section~\ref{secExampleBlackScholes}, where for a highly incomplete market we can narrow the range of admissible prices of a derivative to its analytic price. In general, additional assets could be added to reduce the number of valid risk-neutral measures, and thus the range of admissible prices.

In the following, we discuss the maximizing program Eq.~(\ref{eqLPmax}), as the minimizing program is treated equivalently. 
We make explicit the LP formulation of the problem in Eq.~(\ref{eqLPmax}). 
The first step is to replace the open set $\mathcal Q$ with the closed set $\overline{\mathcal Q}$. In that sense, we do not find probability measures that are strictly equivalent to $\mathbbm P$. However, the probability measures in $\overline{\mathcal Q}$ are arbitrarily close to the equivalent probability measures in ${\mathcal Q}$.  
We use vector notations and the price system to write the linear program
\begin{equation}\label{eq:simplexonthis}
  \begin{array}{ll@{}ll}
  \max\limits_{\bm q \in \mathbbm R^K}  & \bm{D}^T \bm{q}  &\\
  \text{s.t.}& \bm q \geq 0 & &  \\
    & S\bm{q} = \bm{\Pi}. & &
\end{array}
  \end{equation}
From Appendix~\ref{sec:conversionprimaldual}, we see that the dual of the maximizing program is given by a linear program for the vector $\xi \in \mathbbm R^{N+1}$ as
\begin{equation} \label{eqDualMin}
  \begin{array}{ll@{}ll}
  \min\limits_{\xi \in \mathbbm R^{N+1}}  & \bm \Pi^T \bm \xi &\\
  \text{s.t.}& \bm D \leq S^T \bm \xi.  & &
  \end{array}
\end{equation}
The financial interpretation is the following. The dual program finds a portfolio of the traded assets  $\bm \xi$ which minimizes today's price $\bm \Pi^T \bm \xi$. The portfolio attempts to replicate (or ``hedge'') the derivative future payoff. Ideal replication would mean $\bm D = S^T \bm \xi$, i.e, the payoff of the portfolio of assets is exactly the derivative outcome for all market outcomes. The linear program attempts to find a portfolio $\xi$ that ``superhedges" the derivative in the sense that $\bm D \leq S^T \bm \xi$. 

In the following Lemma we show how to recast the pricing LP in standard form (as per Definition \ref{def:lpstandardform} in the Appendix). Note that we also require the parameters defining the LP to assume values in the interval $[-1, 1]$. For this, we need to assume the knowledge of the following quantities.

\begin{assume}[Assumptions for conversion to standard form LP]\label{assumption:dsmax}
Assume the following.
\begin{itemize}
\item
Knowledge of $D_{\max} := \max\limits_\omega D_\omega$, the maximum payoff of the derivative under any event.
\item
Knowledge of $S_{\max} := \max\limits_{ij} S_{ij}$, the maximum entry of the payoff matrix $S$. 
\end{itemize}
\end{assume}

\begin{lemma}[Normalized standard form LP for Equation~\ref{eq:simplexonthis}]\label{lem:stdformLP}
Given Assumption \ref{assumption:dsmax},
the optimization problem in Eq.~(\ref{eq:simplexonthis}) can be written as a linear program in standard form as
\begin{equation} \label{eq:OfOurLP}
  \begin{array}{ll@{}ll}
  \min\limits_{\bm q }  & (\bm D')^T \bm q &\\
  \textnormal{s.t.}& A\bm q \leq \bm c  & \\
  & \bm q \geq 0,
  \end{array}
  \end{equation}

for $A \in [-1,1]^{(2N+2) \times K}$, $\bm D' \in [0,1]^K$, and $\bm c \in [-1,1]^{(2N+2)}$.
\end{lemma}
\begin{proof}
First, use the assumption to normalize the derivative $D$ and the payoff matrix $S$. Define $D'_\omega:= - D_\omega/D_{\max} \in [0,1]$, and $S' := S/S_{\max} \in [0,1]^{N+1 \times K}$ and $\bm \Pi' := \bm \Pi/S_{\max} \in [0,1]^{N+1}$, for which $S'  \bm q = \bm \Pi'$.
Then, we split the equality constraint $S' \bm q = \bm \Pi'$ into two inequality constraints $S' \bm q \leq \bm \Pi'$ and
$(-S') \bm q \leq (-\bm \Pi')$. We can now write the extended problem  
\be
A:=\left(\begin{array}{c}S' \\-S'\end{array} \right)  \bm q \leq \left(\begin{array}{c} \bm \Pi' \\-\bm \Pi' \end{array} \right) =: \bm c.
\ee
\end{proof}
Note that scaling $A$ and $\bm \Pi$ by $S_{\rm max}$ does not change the value of the optimal solution of the LP, but we need to multiply the solution of the LP by $D_{\rm max}$ to obtain the result of the problem in Eq.~(\ref{eq:simplexonthis}). As before, the dual of this problem is given by a linear program for the vector $\xi \in \mathbbm R^{2N +2}$ as
\begin{equation}\label{eq:dualOfOurLP}
\begin{array}{ll@{}ll}
\max\limits_{\bm \xi \mathbbm R^{2N +2}}  & \bm c^T \bm \xi &\\
\text{s.t.}& \bm D' \leq A^T \bm \xi, & \\
& \bm \xi \geq 0.
\end{array}
\end{equation}
The financial interpretation is similar to the above. 
The portfolio $\bm \xi$ is constrained to have only non-negative allocations, but we have assets to invest in with prices $\bm \Pi$ and $- \bm \Pi$, which allows taking long and short positions. We again find the minimum price of the portfolio that also ``superhedges" the derivative via $\bm D' \leq A^T \bm \xi$.

\section{Quantum zero-sum games algorithm for risk-neutral pricing}\label{sec:zerosumgames}

A matrix game is a mathematical representation of a game between two players (here called Alice and Bob) who compete by optimizing their actions over a set of strategies. For such a game, let the possible actions of the two players be indexed by $[m]$ and $[n]$, respectively, and, in the zero-sum setting, let the game be defined by a matrix $F\in [-1,1]^{m \times n}$. The players select the strategy $i \in [m]$ and $j\in [n]$, and the outcome of a round is $F_{ij}$ for the first player and $-F_{ij}$ for the second. One player's gain is the other's loss, hence the name zero-sum game. The randomized (or ``mixed") strategies for the players are described by a vector in the simplex $\Delta^m$ for Alice and $\Delta^n$ for Bob, i.e., probability distributions over the actions. If $\bm x \in \Delta^m$ is the strategy of Alice and $\bm y \in \Delta^n$ is the strategy for Bob, the expected payoff for Alice is $\bm x^TA\bm y$. Jointly, the goal of the players can be formalized as a min-max problem as: 
 \be 
 \min\limits_{\bm y \in \Delta^n} \max\limits_{\bm x \in \Delta^m} \bm x^TA\bm y.
 \ee
The optimal value of a game is defined as 
\be
\lambda^\ast &:=& \min \{ \lambda \in \mathbb{R} | A^T\bm x \leq \lambda \bm 1, \bm x \in \Delta^n \} \nonumber \\
&=&  \max \{ \lambda \in \mathbb{R} | A\bm y \leq \lambda \bm 1, \bm y \in \Delta^m \}.
\ee
This optimal value is always achieved for a pair of optimal strategies $(\bm x,\bm y)$. A strategy is called $\epsilon$-optimal if $A\bm y \leq (\lambda^\ast + \epsilon )\bm 1$. There is a bijection between zero-sum games and linear programs. Following Ref.~\cite[Lemma 12]{Apeldoorn2019}, we can reduce the LP of Lemma \ref{lem:stdformLP} to a zero-sum game. In our scenario, the optimal solution is in $[0,1]$ as $D'_\omega\in [0,1]$ and $\sum_\omega q_\omega = 1$ by construction as a probability measure. 
The problem is solved by binary search on a feasibility problem obtained from the LP. The feasibility problem is deciding if the optimal solution is $< \alpha$ or $\geq \alpha +\epsilon$, for a parameter $\alpha\in[-1,1]$. Note that this interval, for our problem, is restricted to $[0,1]$, as the normalized value of a derivative is non-negative. After manipulating the LP, one can formulate the decision problem solved at each round of the binary search. Below, we state a version of Ref.~\cite[Lemma 12]{Apeldoorn2019}, made specific to our problem. 

\begin{lemma}[Zero-sum game for derivative pricing]\label{thmZeroSumFinance}
Consider the LP in standard form from Lemma~\ref{lem:stdformLP}, with optimal value OPT and optimal strategy $\bm{q^\ast}$, and its dual in Eq.~(\ref{eq:dualOfOurLP}) with optimal solution $\bm \xi^\ast$. Let $R, r >0$ be such that $\sum_{i=1}^K q^\ast_i \leq R$ and $\sum_{j=1}^{N+1} \xi^\ast_j \leq r$.
For an $\alpha \in [0,1]$, the primal LP can be cast as the following LP
\begin{equation}\label{eqZeroSum}
\begin{array}{ll@{}ll}
\max\limits_{\lambda \in \mathbbm R}  & \lambda \\
\text{s.t.}& F \bm y \leq \lambda \bm 1 \\
& \bm y \in \Delta^{N_2},
\end{array}
\end{equation}
where the matrix $F \in [-1,1]^{N_1 \times N_2}$ is defined as
\begin{align} \label{eqMatrixF}
F:=\left(\begin{array}{cccc}
\bm 1^T & 1 & -1\\
-\bm 1^T & 1 & 1\\
 - (\bm{D'})^T & 0 & \alpha/R \\
A & \bm{0} & - \bm{c}/R
\end{array} \right) = \left( 
\begin{array}{cccc}
\bm 1^T & 1 & -1\\
-\bm 1^T & 1 & 1\\
 - (\bm{D'})^T & 0 & \alpha/R \\
S' & \bm{0} & - \bm{\Pi'}/R \\
-S' & \bm{0} &  \bm{\Pi'}/R
\end{array}
\right),
\end{align}
where $N_1:=2(N+1)+5$ and $N_2:= K+2$, and $\bm{0}$ are column vectors of the appropriate size. Finding the optimal value $\lambda^*$ of the LP in Eq.~(\ref{eqZeroSum}) with additive error $\epsilon/(6R(r+1))$ suffices to correctly conclude either $OPT< \alpha$ or $OPT \geq \alpha -\epsilon$ for the original LP.
\end{lemma}

In the previous decision problem, if $\sum_{i=1}^K q^\ast_i \leq R$ then adding the constraint $\sum_{i=1}^K q_i \leq  R$ to the primal in Lemma \ref{lem:stdformLP} will not change the optimal value of the primal. 
In our case, we know that 
\be 
R=1.
\ee
Likewise, if $\sum_{i=1}^{N+1} \xi_i^\ast \leq r$ then adding the constraint $\sum_{i=1}^{N+1} \xi_i \leq r$ will not change the optimal value of the dual program. The LP in Eq.~(\ref{eqZeroSum}) can be solved via a zero-sum game algorithm, which we choose to be the quantum algorithm from Ref.~\cite{Apeldoorn2019}. Solving this linear program obtains a solution vector $\bm y \in \mathbb{R}^{K+2}$, a number $\alpha$, and a number $\widetilde{\lambda}$ from which information on the actual price of the derivative can be obtained.
We have two cases: If $\lambda > \epsilon$, then the option price is $\Pi_{D, \rm max} < \alpha D_{\max}$.
If $\lambda \leq \epsilon/(6(r+1))$, then the option price is $\Pi_{D, \rm max} \geq (\alpha-\epsilon) D_{\max}$. Performing a binary search on $\alpha$ will suffice to estimate the price of the derivative to the desired accuracy. 

\subsection{Quantum input model}
We discuss the input model of our quantum algorithm: an oracle for the price system, and a way of obtaining quantum sampling access to probability distributions. We assume access to the price system and the derivative in the standard quantum query model~\cite{Berry2012} as follows.
\begin{oracle} [Zero-sum game pricing oracles]\label{oracle:query_oracle}
For a price system as in Definition~\ref{def:payoff-matrix}
and a derivative as in Definition~\ref{def:derivative}, 
let $j\in[N_1]$ and $k\in[N_2]$, assume unitaries that at unit cost perform $\ket{jk} \ket 0 \to \ket{jk} \ket {S_{jk}}$,
$\ket{j} \ket 0 \to \ket{j} \ket {\Pi_{j}}$, and $\ket{k} \ket 0 \to \ket{k} \ket {D_{k}}$,
where all numbers are encoded to fixed-point decimal precision.
\end{oracle}
In most practical cases such as conventional European options, the query access to the derivative $\bm D$ can be obtained from a quantum circuit computing the function used to compute the derivative and query access to the payoff matrix $S$.
We also define quantum vector access as a black box that builds quantum states proportional to
any vector (in our context $\bm x$ and $\bm y$). 
\begin{defn}[Quantum vector access]\label{def:quantumsamplingaccess}
Let $\bm x \in \mathbb{R}^d$ with sparsity $s$, and $\|\bm x \|_1 \leq \beta$. We say that we have quantum vector access to $d$-dimensional vectors if for any $\bm x$ we can construct a unitary in time $\Ord{s\ {\rm poly}  \log d}$ that performs the mapping
$$\ket{0}\mapsto \sum_{i=1}^d \sqrt{x_i}/\beta \ket{0} + \ket{G}\ket{1},$$ in $\Ord{ {\rm poly} \log d}$ time, 
where $\ket{G}$ is an sub-normalized arbitrary quantum state.   
\end{defn}

There are different ways of building this quantum access. Given quantum query access to the entries of the vector $\bm x$, we can perform the controlled rotation
$\sum_{i=1}^d \ket{i}\ket{x_i}\ket{0} \mapsto \sum_{i=1}^d \ket{i}\ket{x_i}\left (\sqrt{x_i/\beta} \ket{0}  + \sqrt{1- x_i/\beta} \ket{1}\right)$.
Moreover, the state $\ket{\bm x}= \frac{1}{ \|x\|_1} \sum_{i=1}^d \sqrt{x_i}\ket{i}$ with negligible garbage state can be constructed efficiently using pre-computed partial sums and the quantum circuits developed in \cite{Grover2002, kerenidis2020quantum, Kerenidis2017}. We refer to Definition~\ref{def:sampling_oracle} for a definition of this type of data structure that can be used with a single vector.  In \cite{araujo2021divide} there is a $\Ord{\lceil \log d \rceil^2}$ depth circuit for creating arbitrary quantum states of the form $\sum_{k=0}^{2^n-1}\alpha_k\ket{k}\ket{\rm{garbage}_k}$ where $\ket{\rm{garbage}_k}$ is a $\Ord{2^n}$ qubit state entangled with the first register \cite{sun2021asymptotically}. The quantum access to a vector can be implemented by a quantum random access memory (QRAM)~\cite{Lloyd2013,Giovannetti2008,Giovannetti2008_2,Martini2009}.
A QRAM requires quantum switches that are arranged as a branching tree of depth $\Ord{\log N_1 + \log N_2}$ and width $N_1 \times N_2$ to access the memory elements. As shown in \cite{Giovannetti2008,Giovannetti2008_2}, the expected number of switches that are activated in a memory call is
$\Ord{\log N_1 +\log N_2}$ even though all $\Ord{N_1 N_2}$ switches participate to a call in parallel. More discussion on the architecture can be found in \cite{Arunachalam2015,hann2019hardware}.

\subsection{Quantum algorithm}
\begin{figure}
\begin{algorithm}[H]
  \caption{Quantum LP pricing via zero-sum games \cite{Apeldoorn2019}}
    \label{algGK}
  \begin{algorithmic}[1]
    \Require{
    \Statex Access to: $N+1$ asset prices $\bm \Pi\in \mathbbm R_+^{N+1}$, payoff matrix $S\in \mathbbm R_+^{N+1\times K}$, and derivative $\bm D \in \mathbbm R_+^K$ with Oracle~\ref{oracle:query_oracle}.
    \Statex Knowledge of $S_{\max} = \max_{i,\omega} S_{i\omega}$, $D_{\max} = \max_{\omega} D(\omega)$, and $r >0$. 
    \Statex Access to $\bm x \in \mathbb{R}^{2(N+1)+5}$ and $\bm y \in \mathbb{R}^{K+2}$ with Definition~\ref{def:quantumsamplingaccess}. 
    \Statex Accuracy $\epsilon$, and failure probability $\delta$.}
    \State $x_{\rm l} = -1, x_{\rm r} = 1$ \Comment{Set initial parameters of binary search}
    \State $N_1 = 2(N+1)+5, N_2=K+2$ 
    \State $\alpha=\text{unif}([0,x_{\rm r}])$ 
    \State $\epsilon'=\epsilon/(6(r+1))$ 
    \State Build query access to $F$, as defined in Eq.~(\ref{eqMatrixF}). 
    \For {$1 \textrm{ to } \lceil \log(1/\epsilon) \rceil$} \Comment{Binary search for satisfiability of ZSG}
    \State $\bm x \gets \bm 0 \in \mathbbm R^{N_1}$, $\bm y \gets \bm 0 \in  \mathbbm R^{N_2}$ 
    \State Create $\ket{\bm{p}^{(1)}} = \frac{1}{\sqrt{N_2}}\sum_{k=1}^{N_2}\ket{k}$, $\ket{\bm{q}^{(1)}} = \frac{1}{\sqrt{N_1}}\sum_{k=1}^{N_1}\ket{k}$
    
    \State $T \gets \lceil 16 \epsilon'^{-2}\log {\frac{N_1 N_2\lceil \log(1/\epsilon) \rceil}{\delta} } \rceil$
    \State $\eta \gets \frac{\epsilon'}{4}$
    \For{$t \gets 1 \textrm{ to } T$} \Comment{ZSG algorithm}
    \State Sample $k \in [N_2]$ by measuring $\ket{\bm{p}^{(t)}}$ in the computational basis
    \State Sample $k' \in [N_1]$ by measuring $\ket{\bm{q}^{(t)}}$ in the computational basis
    \State $x_{k'} \gets x_{k'} + \eta$
     \State $y_k \gets y_k + \eta$
         \State Update quantum sampling access for $\bm x, \bm y$.    
     \State Create quantum states (as in Eq.~(\ref{eqQuantumProb})):
     $$\ket{ \bm p^{(t+1)}} = \sum_{k=1}^{N_2} \sqrt{\bm p^{(t+1)}_k} \ket k  \quad\text{ and }\quad \ket{ \bm q^{(t+1)}} = \sum_{k=1}^{N_2} \sqrt{\bm q^{(t+1)}_k} \ket k  $$
    \EndFor
    \State Estimate $\widetilde{\lambda} = x^TFy$  with error $\epsilon'$. \label{line:estimatelambda}
    \If{$\widetilde \lambda> \epsilon'$} \Comment{$\text{OPT} < \alpha$ } \label{line:condizione1}
    
    \State $x_r \gets \alpha; \alpha \gets (x_l + \alpha)/2 $
 \ElsIf{$\widetilde{\lambda} \leq \epsilon'$} \label{line:condizione2} \Comment{$\text{OPT} \geq \alpha -\epsilon$}
        \State $x_l \gets \alpha; \alpha \gets (\alpha + x_r)/2$
\EndIf 
    \State Update quantum access to $F$ with new value of $\alpha$. \label{line:updateF}
    \EndFor
    \State \Return{$\alpha$, $\alpha  D_{\max}$, $\frac{\bm x_{3, \dots, 3+N }}{\|x_{3, \dots, 3+N }\|_1},\frac{\bm y_{0, \dots, K}}{\|y_{0, \dots, K}\|_1}$ }  
  \end{algorithmic}
\end{algorithm}
\end{figure}

The quantum algorithm for zero-sum games uses amplitude amplification and quantum Gibbs sampling to achieve a speedup compared to the classical algorithm. We discuss the run time of the classical algorithm for solving LP based on the reduction to a zero-sum game at the end of this section. 
The quantum computer is used to sample from certain probability distributions created at each iteration. 
The result \cite[Lemma 8]{Apeldoorn2019} is used to prepare the quantum states
\be \label{eqQuantumProb}
\ket{ P^{(t)}} = \frac{1}{\sqrt{\|e^{F\bm y}\|_1}}\sum_{k=1}^{N_2} \sqrt{(e^{F\bm y})_k} \ket{k}, \quad \ket{ Q^{(t)}} &=& \frac{1}{\sqrt{\|e^{-F^T\bm x}\|_1}}\sum_{k=1}^{N_1} \sqrt{(e^{-F^T\bm x})_k} \ket{k},
\ee
via amplitude amplification and polynomial approximation techniques, relying on the oracle access to the matrix $F$. 
Measuring these quantum states then retrieves classical bit strings corresponding to an index $k$ with probability $\bm{P}^{(t)}_k / \Vert \bm{P}^{(t)} \Vert_1$ and $\bm{Q}^{(t)}_k / \Vert  \bm{Q}^{(t)} \Vert_1$, respectively. 
These samples are then used in the steps $y_k \gets y_k + \eta$ and $x_k \gets x_k + \eta$, for classically updating the solution vector. 
The rest of the algorithm remains classical.
We state here the main result from \cite{Apeldoorn2019}. 
An $\epsilon$-feasible solution to our LP problem means that we can find a solution such that the constraint is relaxed to $A\bm q \leq \bm c + \bm 1 \epsilon$.
\begin{theorem}[Dense LP solver~{\cite[Theorem 13]{Apeldoorn2019}}]\label{thm:densesolver}
 Assume to have suitable quantum query access to a normalized LP in standard form as in Lemma~\ref{lem:stdformLP}, with $r,R$ (defined in Lemma~\ref{thmZeroSumFinance}) known, along with quantum vector access as in Definition~\ref{def:quantumsamplingaccess} to two vectors $\bm x \in \mathbb{R}^{N_1}, \bm y \in \mathbb{R}^{N_2}$. For $\epsilon \in (0,1)$, there exists a quantum algorithm that finds an $\epsilon$-optimal and $\epsilon$-feasible $\bm y$ with probability $1-\delta$ using $\tOrd{(\sqrt{K}+\sqrt{N}) \left (\frac{R(r+1)}{\epsilon} \right)^3}$ quantum queries to the oracles, and the same number of gates.
\end{theorem}
Theorem~\ref{thm:densesolver}, can be directly adapted to the martingale pricing problem to obtain a quantum advantage for the run time, see Algorithm \ref{algGK}.  Theorem~\ref{thm:densesolver}, and the embedding in  Lemma~\ref{thmZeroSumFinance} imply the following theorem.  Here, we discuss how to obtain a relative error for the derivative price.

\begin{theorem}[Quantum zero-sum games algorithm for martingale pricing]\label{thm:pricingzsg}
Let $(\bm \Pi, S)$ be a price system with payoff matrix $S \in \mathbb{R}^{(N+1) \times K}$ and price vector $\bm \Pi \in \mathbb{R}^{N+1}$, and let $\bm D\in \mathbb{R}^{K}$ be a derivative. Assume to have quantum access to the matrix $S$ and the vectors $\bm D$ and $\bm \Pi$ through Oracle~\ref{oracle:query_oracle}, and two vectors $\bm x \in \mathbb{R}^{N_1}$, and $\bm y \in \mathbb{R}^{N_2}$ through 
quantum vector access as in Definition~\ref{def:quantumsamplingaccess}. Under Assumption~\ref{assumption:dsmax}, for $\epsilon \in (0,1/2)$, Algorithm~\ref{algGK} estimates $\Pi_{D, \max}$ with absolute error $\epsilon$ and with high probability using $$\Ord{(\sqrt{K}+\sqrt{N}) \left (\frac{(r+1)}{\epsilon }D_{\max}\right)^3}$$ 
quantum queries to the oracles, and the same number of gates. The algorithm also returns quantum access to an $\epsilon$-feasible solution $\bm q$ as in Definition~\ref{def:quantumsamplingaccess}.
\end{theorem}
\begin{proof}
From quantum access to the vectors $\bm D$ and $\bm \Pi$, and knowing $D_{\rm max}$ and $S_{\rm max}$ we can obtain quantum access to $\bm D'$ and $\bm \Pi'$ using quantum arithmetic circuits (see for example \cite{ruiz2017quantum}). From the given oracles and quantum arithmetic circuits, we can derive quantum access to an oracle for the matrix $F$ in Eq.~(\ref{eqMatrixF}). It follows from these observations that one query to the entries of $F$ costs $\Ord{1}$ queries to the oracles in the input. Note that the updating query access to the matrix $F$ (Line~\ref{line:updateF}) takes constant time, as we do not need to modify the whole matrix $F$.

In order to estimate the value of the zero-sum game (Line~\ref{line:estimatelambda}) with $\epsilon$ error and probability greater than $1-\delta$ we use classical sampling. As stated in \cite[Claim 2]{Apeldoorn2019}, we require $k = \Ord{\log(1/\delta)/(\epsilon')^{2}}$ independent samples $i_1, \dots, i_k$ from 
$\bm x$ and similarly $j_1, \dots, j_k$ independent samples from $\bm y$. Then, $\sum_k A_{i_k, j_k}$ is an estimate $\widetilde{\lambda}$ of $\lambda^*$ with absolute error $\epsilon'$. As each repetition of the ZSG algorithm does not return the right value of the game $\lambda^\ast$, but only an estimate with absolute error $\epsilon'$, the condition on $\epsilon$ for the precision of the binary search of \cite[Lemma 12]{Apeldoorn2019} has to be modified to account for this case. This leads to the conditional statements in Line~\ref{line:condizione1} and Line~\ref{line:condizione2}. 
\end{proof}
Note that the original algorithm of Theorem~\ref{thm:densesolver} results in an absolute error $\epsilon$ for $\alpha$, for the normalized problem, which then results in an absolute error of $\epsilon D_{\rm max}$ for the value of the derivative.

\begin{theorem}[Quantum zero-sum games algorithm for martingale pricing, relative error]\label{thm:pricingzsg_rel}
In the same setting as Theorem \ref{thm:pricingzsg}, Algorithm~\ref{algGK} estimates $\Pi_{D, \max}$ with relative error $\epsilon$ and with high probability using $$\tOrd{(\sqrt{K}+\sqrt{N}) \left (\frac{(r+1)}{\epsilon }\rho\right)^3}$$ 
quantum queries to the oracles, and the same number of gates, where $\rho=D_{\max}/\Pi_{D, \rm max}$. 
\end{theorem}
\begin{proof}
In order to achieve a relative error $\epsilon$ we can run a standard procedure. For an integer $k$, run Algorithm \ref{algGK} with precision parameter set to $\frac{1}{2^k}$ and success probability at least $0.99$. For a given $k$, this will return an estimate ${\alpha}_k$ such that $|\alpha - \alpha_k| \leq \frac{1}{2^k}$ with success probability at least $0.99$.
Next, conditioned on the event that the algorithm succeeds, check the estimate for $\alpha_k -\frac{3}{2^k} > 0$. If that does hold, we re-run the algorithm with precision $\epsilon (\alpha_k-\frac{1}{2^k})$ and success probability $0.99$.  This obtains an estimate $\widetilde{\alpha}_k$ such that
$|\alpha - \widetilde{\alpha}_k| \leq \epsilon (\alpha_k -\frac{1}{2^k})
\leq \epsilon(\alpha + \frac{1}{2^k}-\frac{1}{2^k}) \leq \epsilon \alpha$ with success probability $0.98$ by the union bound. For the run time, we evaluate
$\alpha_k - \frac{1}{2^k} = \frac{\alpha_k}{2} + \frac{\alpha_k}{2} - \frac{3}{2^{k+1}} +\frac{1}{2^{k+1}} 
> \frac{\alpha_k}{2}+\frac{1}{2^{k+1}} \geq \frac{\alpha}{2} -\frac{1}{2^{k+1}} +\frac{1}{2^{k+1}}=\frac{\alpha}{2}$,
and hence $1/(\alpha_k - \frac{1}{2^k}) < 2/\alpha$.
Set $k^\ast := k$ and $\widetilde \alpha := \widetilde{\alpha}_{k^\ast}$.
If the above criterion does not hold, we increase $k \to k+1$. The procedure will halt because by assumption $\alpha>0$. The number of iterations is given by the first $k$ such that $0< \alpha - \frac{3}{2^k}$, which is bounded by $\Ord{\log(1/\alpha)}$. 
By repeating enough times we can make sure that the total success probability is high.
We obtain the estimate $\widetilde{\Pi}_{D, \max} := \widetilde{\alpha} D_{\max}$ with relative error in the price $\vert \Pi_{D, \max} - \widetilde{\Pi}_{D, \max}\vert \leq \epsilon \Pi_{D, \max}$ with high probability and with a run time
$\tOrd{(\sqrt{K}+\sqrt{N}) \left (\frac{(r+1)D_{\max}}{\epsilon \Pi_{D,\max} }\right)^3}$.
\end{proof}
Along with the price with relative error, this algorithm returns also a $\epsilon$-feasible solution $\widetilde{\bm q}$. This solution could be used to price other derivatives $\bm D$. Doing so will often result in an
error as large as $| \widetilde{\bm{q}}^T\bm{D} - \bm q^T \bm D | \leq \epsilon K D_{\max}$.

\paragraph{Classical zero-sum games algorithm and condition for quantum advantage} 
The classical algorithm requires a number of iterations  $T=\Ord{\left(\frac{(r+1)}{\epsilon}\rho \right)^2 \log(\frac{NK}{\delta}) }$ to achieve the relative error $\epsilon$ for $\Pi_D$.
The cost of the computation for a single iteration is dominated by the update of the two vectors $\bm x$ and $\bm y$ of size $\Ord{N}$ and $\Ord{K}$, respectively. Hence, the classical run time of the algorithm is
$\Ord{T (N_1+N_2)} = \Ord{\frac{(N+K)(r+1)^2 }{\epsilon^2}\rho^2\log {\frac{NK}{\delta}}}$.

Now we focus on a condition that allows for the quantum algorithm to be faster than the classical counterpart. The run time of the algorithms depends on the parameter $r$, which we can relate to the number of assets required to hedge a derivative. 

\begin{remark} \label{remark:upperboundfor_r}
Consider the context
of Lemma \ref{lem:stdformLP} and its dual formulation.
For any vector $\bm \xi \in \mathbbm R^{N+1}$, note that $\sum_{j=1}^{N+1} \xi_j \leq \Vert \bm \xi\Vert_1$. In addition, if $\xi_j \leq 1$, we have $\Vert \bm \xi\Vert_1 \leq \Vert \bm \xi\Vert_0$. In our setting, where $\bm \xi^\ast$ is the optimal portfolio and $\Vert \bm \xi^\ast\Vert_0$ is the number of assets in the optimal portfolio, we hence can use $r \leq \Vert \bm \xi^\ast\Vert_0$.
\end{remark}

To obtain an advantage in the run time of the quantum algorithm we have to presuppose the following for $\bm \xi$.
\begin{theorem}\label{lemma:speedupzsg}
If the derivative requires $\|\bm \xi^\ast\|_0 \leq \frac{\epsilon(\sqrt{N+K})}{\rho}$ asset allocations to be replicated, where $\rho=
D_{\max}/\Pi_{D, \rm max}$, then Algorithm \ref{algGK} performs less queries to the input oracles compared to the classical algorithm.
\end{theorem}
\begin{proof}
Since the number of iterations is the same, we focus on the query complexity of a single iteration. To find the upper bound for $r$ which up to logarithmic factors still allows for a speedup, we need to find $r$ such that 
$ (\sqrt{N}+\sqrt{K})\left( \frac{(r+1)} {\epsilon }\rho\right)^3 < \frac{(N+K) (r+1)^2 }{\epsilon^2}\rho^2 $.  
Thus,
\be
r+1 \leq \epsilon\  \frac{ (N+K)}{\rho(\sqrt{N}+\sqrt{K})} = \epsilon\   \frac{(\sqrt{N+K})(\sqrt{N+K})}{\rho(\sqrt{N}+\sqrt{K})} \leq & \epsilon\ \Pi_D  \frac{(\sqrt{N+K})(\sqrt{N}+\sqrt{K})}{\sqrt{N}+\sqrt{K}} = \frac{\epsilon(\sqrt{N+K})}{\rho}. 
\ee
By Remark~\ref{remark:upperboundfor_r}, we have an advantage in query complexity if the number of assets in the optimal portfolio $\|\bm \xi^\ast\|_0$ is bounded by $\frac{ \epsilon \sqrt{N+K}}{\rho}$.
\end{proof}

\subsection{
Linear programming martingale pricing with the Black-Scholes-Merton model}
\label{secExampleBlackScholes}
We show a simple example motivated by the standard Black-Scholes-Merton (BSM) framework. In our example the payoff matrix can be efficiently computed, and hence  access to large amounts of additional data is not required. The consequence is that instead of having quantum access to $S$ directly, for example from a QRAM of size $O(KN)$, we only require access to data of size $O(N)$. Moreover, we show an example of how an additional constraints in the LP can regularize the model, so to narrow the range of admissible price. For a call option, which we consider in this example, the price can be computed analytically. We show how to make the range of prices of the solutions of our linear programs comparable to the analytic value of the derivative.

In the BSM model, the market is assumed to be driven by underlying stochastic processes (factors), which are multidimensional Brownian motions. Multi-factor models have a long history in economics and finance \cite{Bai2003}.
The dimension and other parameters of the model, like drift and volatility, can be estimated from past market data \cite{bjork2009arbitrage,mayhew1995implied,berestycki2004computing}. 
An advantage of modeling the stocks with the BSM model is that it often allows to solve analytically the pricing problem for financial derivatives.

\paragraph{Pricing with a measure change}
An importation notion in pricing is the \emph{change} of measure from the original measure to the martingale measure. 
The change of measure is described by the Radon-Nikodym derivative, which  denoted by the random variable  $X:=\frac{d\mathbbm Q}{d\mathbbm P}$.
We formulate the problem of pricing a derivative, given that we have an initial probability measure $\mathbbm P$ via a vector $\bm p \in \Delta^K$ and consider the measure change to a new probability measure $\mathbbm Q$ via the vector $\bm q \in \Delta^K$. 
The Radon-Nikodym derivative in this context is an $\bm x \in \mathbbm R^K_+$ for which
\be
x_\omega := \frac{q_\omega}{p_\omega}.
\ee
The martingale property for the assets under $\mathbbm Q$ and the pricing of a derivative
are formulated as
\be
\Pi_j = \mathbb{E}_{\mathbbm Q}[S_j] &=& \mathbb{E}_{\mathbbm P}[X S_j]  = \sum_{\omega \in \Omega} p_\omega x_\omega S_{j\omega},\\
\mathbb{E}_{\mathbbm Q}[D] &=& \mathbb{E}_{\mathbbm P}[X D] = \sum_{\omega \in \Omega} p_\omega X_\omega D_\omega.
\ee
Hence, using $\circ$ for the Hadamard product (element-wise product) between vectors,  $\bm D[\bm p] := \bm p \circ \bm D$ and $(S[\bm p])_{j} := \bm p\circ S_{j}$,
we obtain the linear program for the maximization problem of Eq.~(\ref{eqLPmax}) of the Radon-Nikodym derivative
\begin{equation}\label{eqRadonLP}
  \begin{array}{ll@{}ll}
  \max\limits_{\bm x \in \mathbbm R^K}  & \bm{D}[\bm p]^T \bm{x}  &\\
  \text{s.t.}& \bm x \geq 0 & &  \\
    & S[\bm p]\bm{x} = \bm{\Pi}. & &
\end{array}
  \end{equation}
The minimization problem is obtained similarly.
  
\paragraph{Discretized Black-Scholes-Merton model for a single time step.} 
We perform our first drastic simplification and set the number of risky assets to be $N=1$. This makes the market highly incomplete but simplifies the modeling.
While the Brownian motion stochastic process has a specific technical definition, in the single-period model a Brownian motion becomes a Gaussian random variable.  
Hence, we assume here that we have a single Gaussian variable driving the stock.
Let $B : \Omega \to \mathbbm R$ be the Gaussian random variable, also denoted by $B\sim \mathcal N(0,1)$ distributed under the measure $\mathbbm P$.
The future stock price $S_2$ is modeled as a function of $B$. 
In BSM-type models, asset prices are given by a log-normal dependency under $\mathbbm P$,
\be
S_2 = \Pi\ e^{\sigma B + \mu -\frac{1}{2}\sigma^2},
\ee
using today's price $\Pi\in \mathbbm R_+$, the volatility $ \sigma \in \mathbbm R_+$, and the drift $\mu\in \mathbbm R_+$. 

The stock price is not a martingale under $\mathbbm P$, which can be checked by evaluating $\mathbbm E_{\mathbbm P}[S] \neq \Pi$.
However, it can be shown that under a change of probability measure, the Gaussian random variable can be shifted such that the resulting stock price is a martingale. 
The following lemma is a simple version of Girsanov's theorem for measure changes for Gaussian random variables.
\begin{lemma}[Measure change for Gaussian random variables \cite{shreve2005stochastic}]\label{lem:grisanovgaussian}  Let $B \sim \mathcal N(0,1)$ under a probability measure $\mathbbm P$. 
For $\theta \in \mathbbm R$ define
$B' := B + \theta$. Then, $B' \sim \mathcal N(0,1)$ under a probability measure 
$\mathbbm Q$, where this measure is defined by $ \frac{d \mathbbm Q}{d \mathbbm P} = \exp\left (-\theta B - \frac{1}{2} \theta^2\right)$.
\end{lemma}
In standard Black-Scholes theory, we use $\theta = \mu/\sigma$ to obtain  $ B' \sim \mathcal N(0,1)$ under $\mathbbm Q$
where 
\be
\frac{d \mathbbm Q}{d \mathbbm P} = \exp\left (-\frac{\mu}{\sigma}  B - \frac{1}{2} \left(\frac{\mu}{\sigma}\right)^2\right).
\ee
We can check that 
$S_2$ satisfies the key property of a martingale as
$\mathbbm E_{\mathbbm Q}[S_2] = \mathbbm E_{\mathbbm P}\left [\frac{d \mathbbm  Q}{d\mathbbm P} S_2\right] = \Pi\ \mathbbm E_{\mathbbm P}[e^{\beta B -\frac{1}{2} \beta^2}] = \Pi$,
with $\beta := \sigma - \mu/\sigma$.
Many financial derivatives can be priced in the Black-Scholes-Merton model by analytically evaluating the expectation value under $\mathbbm Q$ of the future payoff. We use this analytically solvable model as a test case for the linear programming formulation. In order to fit this model into our framework, we need to discretize the sample space. In contrast to the standard BSM model, we take a truncated and discretized normal random variable centered at $0$ on the interval $[-6, 6]$.
We take the sample space to be $\Omega := \{-K_0,\cdots,0,\cdots, K_0\}$ and $K_0 \in \mathbbm Z_+$. The size of the space is $K = \vert \Omega \vert = 2K_0+1$. 
The probability measure for $\omega \in \Omega$, is
$p_\omega = e^{-\frac{36\omega^2}{2 K_0^2}}/ \Sigma$, where $\Sigma = \sum_{\omega \in \Omega} e^{-\frac{36\omega^2}{2 K_0^2}}$.
In analogy to the BSM model, we define the payoff matrix of future stock prices as 
\be 
S_{1\omega} &:=& 1, \nonumber\\
S_{2\omega} &:=& \Pi\ e^{ {\sigma} B(\omega) + \mu + \frac{1}{2}\sigma^2},\label{eqBSMdiscrete}
\ee
where $B(\omega) = \frac{6\omega}{K_0}$.
An important property of this example is that each $S_{2\omega}$, can be computed in $\Ord{\log K}$ time and space via Eq.~(\ref{eqBSMdiscrete}). For settings with more than a single asset $N>1$, if we have more than one, say $d>1$, Gaussian random variables driving the assets, we expect this cost to be $\Ord{d\log NK}$.

\paragraph{Call option.} 
We define the simple case of an European call option. 
The European call option on the risky asset gives the holder the right but not the obligation to buy the asset for a predefined price, the strike price. Hence, the future payoff is defined as
\be
D := \max \left \{0, S_2  - Z \right \},
\ee where  
$Z \in \mathbbm R_+$ is the strike price.
Using the stock prices from Eq.~(\ref{eqBSMdiscrete}), for a specific event $\omega \in \Omega$ the payoff under $\mathbbm P$ is
\be
D(\omega) =  \max \left \{0, \Pi_i e^{ \sigma B(\omega) + \mu -\frac{1}{2}\sigma^2}  - Z \right \}.
\ee

\paragraph{Numerical tests.}

We now test this approach for the call option, using the linear program in Eq.~(\ref{eqRadonLP}), its minimization form, and the payoff matrix Eq.~(\ref{eqBSMdiscrete}). In Appendix~\ref{app:BSMRadon} we present the main numerical results.
We compare the approach to the analytical price derived from the Black-Scholes-Merton model. Our optimization framework allows for a much larger set of Radon-Nikodym derivatives (all entry-wise positive $K$-dimensional vectors) than the set implied from the BSM framework. We observe that this freedom allows for a large range of prices corresponding to idiosyncratic solutions for the measure change. 
To find solutions closer to the Black-Scholes measure change, we develop a simple regularization technique. This regularization is performed by including a constraint for the slope of $\bm q$ into the LP, see Appendix~\ref{app:BSMRadon}.
The regularization leads to a narrowing of the range of admissible prices.
We find that for a broad enough range of parameters, the analytic price of the derivative is contained in the interval of admissible prices given by the maximization and minimization linear programs of Eq.~(\ref{eqRadonLP}).
The regularization is successful in recovering the BSM price for strong regularization parameter. In that sense, this linear programming framework contains the BSM model and allows to model effects beyond the BSM-model.

Instead of using the regularization, there could be other approaches to narrow the range between admissible prices. An approach founded in the theory is to complete the market with assets for which we know the price analytically/from the market.
When there are more benchmark assets, we expect the range of prices to narrow. In the limit of a complete market, every derivative has a unique price.
We leave this approach for future work.

\paragraph{Towards a quantum implementation.}
The requirement for an efficient quantum algorithm is to have query access to oracles (as in Oracle \ref{oracle:query_oracle}) for the vectors.
For the payoff matrix $S$ we can build efficient quantum circuits, without requiring large input data. To create the mapping $\ket{i,j} \mapsto \ket{i,j,S_{ij}}$, we have to have classical access to input quantities such as drift and volatility, quantum access to the prices, and apply the arithmetic operations in Eq.~(\ref{eqBSMdiscrete}). Following the same reasoning, since the vector $\bm D$ is efficiently computable from the matrix $S$, we can build efficient query access to its entries.  We remark that such efficient computation can hold also for more complicated models (i.e., models that are more complicated than log-normal such as Poisson jump processes or Levy heavy-tailed distributions). 

\section{Quantum matrix inversion for risk-neutral pricing}\label{appendix:potentialQLA}
Taking insight from the fact that the sought-after measure $\bm q$ is a vector satisfying the linear system  $\bm \Pi = S \bm q$, we consider in this section the quantum linear systems algorithm (\cite{harrow2009quantum}, and successive improvements \cite{gilyen2019quantum,chakraborty2018power}) for solving the pricing problem. Given the right assumptions about the data input and output, one could hope to obtain an exponential speedup for the pricing problem. However, the quantum linear systems algorithm cannot directly be used to find positive solutions for an underdetermined linear equations system \cite{Rebentrost2014}. The constraint $\bm q \geq 0$ involves $K$ constraints of the form $ q_\omega \geq 0$, which all have to be enforced. We may be able to apply the quantum linear systems algorithm if we demand a stronger no-arbitrage condition, as will be discussed in the present section. With quantum linear systems algorithms we can prepare a quantum state that is proportional to $S^+ {\bm \Pi}$, where $(\bm \Pi, S)$ is a price system and $S^+$ is the pseudo-inverse of the matrix $S$. A well-known fact (which we recall in Theorem~\ref{thm:ell2minimization} in the Appendix) is that the pseudo-inverse of $S$ finds the minimum $\ell_2$ norm solution to a linear system of equations. Unfortunately, this solution can lie outside the positive cone, and thus outside ${\rm int} (\Delta^K)$, so we need a stronger no-arbitrage condition to guarantee finding a valid probability measure. 
We have to assume that one of the valid probability measures in $\mathcal Q$ - whose existence follows from the no-arbitrage assumption of Definition \ref{def:arbitrage} - is indeed also the minimum $\ell_2$-norm solution. A market satisfying this condition we call a least-square market. 

\begin{defn} [Least-squares market]\label{def:assumeL2} 
An arbitrage-free market model is called a least-squares market if $S^+\bm \Pi \in \mathcal{Q}$.
\end{defn}
Under this assumption, the pseudo-inverse is guaranteed to find a positive solution to $S \bm q= \bm \Pi$. 
This assumption is obviously stronger than the no-arbitrage assumption (Assumption~\ref{def:arbitrage}), which only implies that $|\mathcal{Q}|\neq 0$. 
However, it is not as strong as assuming market completeness, see Definition \ref{assumeComplete}, which means that every possible payoff is replicable with a portfolio of the assets. We now state an intermediate result for the rank of $S$. 
\begin{lemma}\label{lem:completeness is fcr}
Let $(\bm \Pi, S)$ be a price system, as in Definition \ref{def:payoff-matrix}. The market model is complete if and only if $N+1=K$ and the matrix $S \in \mathbb{R}^{(N+1) \times K}$ has full rank. 
\end{lemma}

\begin{proof}
If the market is complete, any contingent claim is attainable. For a contingent claim $\bm D \in \mathbbm R_+^K$ let $\bm \xi \in \mathbbm R^{N+1}$ be a replicating portfolio so that $\bm D=S^T\bm{\xi}$. Since $\bm D$ can be any vector in $\mathbb{R}_+^K$, 
    the column space of $S^T$ must have dimension equal to $K$. Since we assume non-redundant assets, it must hold that $N+1 = K$ and $S$ has full column rank. 
Now assume that $S$ has full rank and that the number of assets is equal to the number of events (i.e. $N+1=K$). Having full rank implies that the matrix $S$ is invertible, and thus $(S^T)^{-1} \bm{D} = \bm \xi$ is the unique replicating portfolio.
\end{proof}

From the second fundamental theorem of asset pricing (Theorem~\ref{thm:sftap}), there exists a unique martingale measure if and only if every financial derivative is redundant, i.e., replicable via a unique portfolio of traded assets. If the market is not complete,
a derivative is in general not replicated. The next theorem states that market completeness is a stronger assumption than the assumption of a least-square market.  
\begin{theorem} Let $(\bm \Pi, S)$ be a price system, as in Definition \ref{def:payoff-matrix}. If the market model $S \in \mathbb{R}^{(N+1) \times K}$ is complete, it is also a least-squares market.
\end{theorem}\label{thm:spseudoinvisq}

\begin{proof}
Market completeness, via Theorem \ref{thm:sftap}, implies that $|\mathcal{Q}|=1$, so there exists a unique measure $\bf q^\ast \in \mathcal{Q}$ such that $\bm q^\ast \in {\rm int}(\Delta^K)$ and $S \bm q^\ast = \bm \Pi$.
Market completeness, via Lemma~\ref{lem:completeness is fcr}, implies that the matrix $S$ is square and has full rank.  Define the pseudo-inverse solution as
$\bm q^+:= S^{+} \bm \Pi$.
Since $N+1=K$, and the assets are non-redundant, $\bm q^+$ satisfies $S \bm q^+= \bm \Pi$ and is the unique solution to the equation system. Thus, $\bm q^+$ must also have the property of being in ${\rm int}(\Delta^K)$ and $\bm q^+ = \bm q^\ast$. 
\end{proof}
Market completeness is usually a rather strong hypothesis, while least-square market includes a broader set of markets. 
Both the market completeness and the weaker least-squares market conditions enable the use of quantum matrix inversion for finding a martingale measure.  Even weaker conditions may also suffice for the use of quantum matrix inversion. For instance, markets where some $\bm q \in \mathcal Q$ is $\epsilon$-close to $S^+\bm \Pi$ may be also be good candidates. 
Another case could arise in the situation where there is a vector $\bm v$ in the null space of $S$ such that $(\bm q^+ + \bm v) \in  \mathcal{Q}$.
We leave the exploration of weaker assumptions for future work. 

In this part of the work, we assume block-encoding access to $S$ \cite{gilyen2019quantum}. 
\begin{defn}[Block-encoding]\label{def:block-encoding}
Let $A\in \mathbb{C}^{2^s \times 2^s}$. We say that a unitary $U \in \mathbb{C}^{(s+a)\times(s+a)}$ is a ($\alpha, a, \epsilon)$ block encoding of $A$ if
$$\norm{A - \alpha (\bra{0}^a \otimes I )U( \ket{0}^a \otimes I)  } \leq \epsilon.$$
\end{defn}
Often the shorthand $(\alpha, \epsilon)$-block encoding is used which does not explicitly mention the number of ancillary qubits. An $(\|A\|_F, 0)$-block encoding can be obtained with a variety of access models \cite{kerenidis2020quantum, chakraborty2018power}.
Note that the matrix $S$ is in general not sparse, as most events will not lead to zero stock prices.
To satisfy Definition \ref{def:block-encoding} and the hypothesis of Theorem~\ref{cor:pseudoinv-vtaa}, we use the standard trick of obtaining a symmetric version of $S$ by creating an \textit{extended} matrix $\in \mathbb{R}^{(K+N+1)\times (K+N+1)}$ with $S$ in the upper-right corner and $S^T$ in the bottom-left corner, and $0$ everywhere else. In the following, we assume w.l.o.g that $K$ and $N+1$ are power of $2$ (if they are not, we can always pad the matrix and the vectors with $0$s) and we do not introduce a new notation for the padded quantities. 

Under the assumption that the matrix $S$ can be efficiently accessed in a quantum computer, we can prepare a quantum state for the martingale measure in time $\Ord{\|S\|_F\kappa(S) \log\left(\frac{N+K}{\epsilon}\right)}$ to accuracy $\epsilon$ \cite{harrow2009quantum,Childs2017linear,chakraborty2018power,gilyen2019quantum}. Using the quantum state for the probability measure, one can then price the derivative by computing an inner product with another quantum state representing the derivative. This will result in a run time that will depend linearly in the precision required in the estimate of the value of the derivative and further scaling factors given by the normalization of the vector representing the derivative and the martingale measure $\bm q$ as a quantum state.   

\begin{theorem}[Quantum pseudo-inverse pricing in incomplete least-squares markets]\label{thm:pricingqla}
Let $(\bm \Pi, S)$ be the price system of a least-squares market with the extended payoff matrix $S\in \mathbb{R}^{(K+N+1)\times (K+N+1)}$ satisfying $\|{S}\|_2 \leq 1$ and $\kappa(S) \geq 2$. Let $U_{S}$ be a $(\|S\|_F, 0)$-block-encoding of matrix $S$, and assume to have quantum access as Definition~\ref{def:sampling_oracle} to the extended price vector $\bm \Pi\in \mathbb{R}^{(K+N+1)}$ and the extended derivative vector $\bm D \in \mathbb{R}^{(K+N+1)}$.
Let $\gamma \in [0,1]$ and  $\sqrt{\gamma} \leq \|P_{\rm{col}}(S)\ket{\bm \Pi}\|_2$, where $P_{\rm{col}}(S)$ is the projector into the column space of $S$.
For $\epsilon > 0$ and $\bm q^+ := S^+ \bm \Pi$, there is a quantum algorithm that estimates $\Pi_{D,\bm q^+} := \bm D^T \bm q^+$ with relative error $\epsilon$ and high probability using $\tOrd{\frac{\kappa(S)\|S\|_F}{\epsilon \sqrt{\gamma}} \frac{\| \bm D\|_2}{\Pi_{D, \bm q^+} } }$
queries to the oracles.
\end{theorem}

\begin{proof}
With $\bm q^+ = S^+ \bm \Pi$, the price of the derivative can be rewritten as:
\be
\Pi_{D,\bm q^+} = \mathbbm{E}_{\bm q^+}[{D}] = \sum_{\omega\in\Omega} D_\omega q_\omega = \|\bm q^+\|_2 \|\bm D\|_2 \braket{\bm q^+|\bm D}= \|{\bm \Pi}\|_2 \left \|{S}^+\ket{{\bm \Pi}} \right \|_2\|{\bm D}\|_2 \braket{ \bm q^+|{\bm D}}.
\ee
It is simple to check (using the triangle inequality) that to estimate $\Pi_{D,\bm q^+}$ with relative error $\epsilon$ it suffices to estimate $\braket{\bm q^+|\bm D}$ and $\|{S}^+\ket{{\bm \Pi}}\|_2$  with relative error $\epsilon/2$, as we have exact knowledge of $\|{\bm \Pi}\|_2$ and $\|{\bm D}\|_2$ (due to the availability of Definition~\ref{def:sampling_oracle}). The idea is to use Theorem~\ref{cor:pseudoinv-vtaa} for the ability to prepare the state $\ket{ \bm q^+}:=  S^+\ket{{\bm \Pi}}/\| S^+\ket{{\bm \Pi}}\|_2$, estimate the normalizing factor, and then perform a swap test between $\ket{ S^+{\bm \Pi}}$ and $\ket{{\bm D}}$. Using $U_S$, and Theorem~\ref{cor:pseudoinv-vtaa} we can prepare a state $\ket{\widetilde{\bm q}^+}$ close to $\ket{\bm q^+}$, with a run time poly-logarithmic factor in the precision. 
Let $T_{{S}}$, $T_{{\bm \Pi}}$, and $T_{{\bm D}}$ the cost in terms of elementary gates for having quantum access to ${S}$, ${\bm \Pi}$, and ${\bm D}$. Producing $\ket{\bm q^+}$  with Theorem~\ref{cor:pseudoinv-vtaa} has a cost in terms of elementary gates $T_{\bm q^+}$ of 
\be
T_{\bm q^+} \in \tOrd{\frac{\kappa(S)\Vert S\Vert_F}{\sqrt{\gamma}} \left[ T_{{S}} + \lceil \log(N+K)\rceil  + T_{{\bm \Pi}}\right]}.
\ee
Estimating $\|{S}^+\ket{{\bm \Pi}}\|_2$ with relative error $\epsilon/2$ has a cost of $\tOrd{ \frac{T_{\bm q^+}}{\epsilon}}$ according to
Theorem~\ref{cor:qls-vtae}. 
We now estimate with relative error $\epsilon'$ the value $\braket{\bm q^+ \big | {\bm D}}$  with a standard trick which requires $O\left(   \frac{1}{\epsilon'}\right)$ calls to the unitaries generating $\ket{ \bm q^+}$ and $\ket{{\bm D}}$. 
Prepare the state 
\be
\ket{\psi}=\frac{1}{2}\left(\ket{0}\left (\ket{{\bm D}} + \ket{\bm q^+}\right ) +  \ket{1}\left (\ket{{\bm D}} - \ket{\bm q^+} \right) \right),
\ee
and then use amplitude estimation (Theorem~\ref{thm:ampest}) to estimate the probability of measuring $0$ in the first qubit, which is $p(0)=\frac{2+2\braket{{\bm D}| \bm q^+}}{4}$. 
To obtain a relative error of $\epsilon/2$ on $\braket{{\bm D}| \bm q^+}$ we set $\epsilon'=\epsilon \braket{{\bm D}| \bm q^+}/4$. 
Thus, the number of queries to the unitaries preparing the two states is $\Ord{1/\left(\epsilon \braket{\bm q^+ | {\bm D}}\right)}$.
The total run time is the sum of the two estimations but the inner product estimation has an additional factor of  
$\braket{{\bm D}| \bm q^+}$ at the denominator, and therefore dominates the run time. 
Repeating the algorithm for a logarithmic number of times achieves a high success probability via the union bound. 
The final run time is $\tOrd{\frac{ \kappa(S)\| S\|_F}{\epsilon\braket{{\bm D}| \bm q^+ }\sqrt{\gamma}}}$. 
\end{proof}

Interestingly, when the market is complete, the factor $\sqrt{\gamma}$ is just $1$, so we can think of this parameter as a ``measure'' of market incompleteness, which is reflected in the run time of our algorithm (and can be estimated with amplitude estimation).
In addition, if we want to recover an approximation of the risk-neutral measure $\bm q^+$. Obtaining an estimate of $\bm q^+$ allows to price derivatives classically. 
\begin{theorem}[Estimation of the martingale measure]\label{thm:estimatingqqla}
Consider the same settings as Theorem~\ref{thm:pricingqla}. For $\epsilon >0$, there is a quantum algorithm that returns a classical description $\widetilde{\bm q}^+_{\rm est}$ of $\bm q^+$, such that a derivative can be priced as follows.
Given a derivative $\bm D$ 
we can estimate $\Pi_{D, \bm q^+}$ with relative error $\epsilon$ using 
$\tOrd{K\frac{ \kappa(S)\|{S}\|_F }{\sqrt{\gamma}\epsilon^2} \frac{\| \bm D\|_2^2}{\Pi_{D, \bm q^+}^2 } }$ queries to the oracles.
\end{theorem}
\begin{proof}
Let $\ket{\widetilde{\bm q}^+_{\rm est}}$ be the $\ell_2$-norm unit vector obtained through tomography (Theorem ~\ref{thm:tomography} in the appendix) on the quantum state $\ket{\widetilde{\bm q}^+}$ with $\ell_2$ error $\left \Vert \ket{\widetilde{\bm q}^+_{\rm est}}- \ket{\widetilde{\bm q}^+}\right\Vert_2 \leq \epsilon_2$.
As in Theorem~\ref{thm:pricingqla}, we can take that $\ket{\widetilde{\bm q}^+} \approx \ket{\bm q^+}$ with negligible error. 
Let $\widetilde{\norm{\bm q^+}}$ our estimate of $\norm{\bm q^+}$ with error $\epsilon_1$. Let  $\widetilde{\bm q}^+_{\rm est}:=\widetilde{\norm{\bm q^+}}\ket{\widetilde{\bm q}^+_{\rm est}}$ be the estimate of the vector $\bm q^+$. 
For pricing a derivative, rewrite $ |\bm D^T \bm q^+ - \bm D^T \widetilde{\bm q}^+_{\rm est}| \leq \bm D^T \bm q^+ \epsilon = \| \bm D\|_2 \|\bm q^+\|_2 \braket{\bm D | \bm q^+}\epsilon$.
Using Cauchy-Schwarz, we can also bound $|\bm D^T \bm q^+ - \bm D^T \widetilde{\bm q}^+_{\rm est}| = |\bm D^T(\bm q^+ - \widetilde{\bm q}^+_{\rm est}) | \leq \|\bm D\|_2 \|\bm q^+ - \widetilde{\bm q}^+_{\rm est}\|_2$. Thus, we want to find $\epsilon_1$ and $\epsilon_2$ such that $\|\bm q^+ - \widetilde{\bm q}^+_{\rm est}\|_2 \leq \|\bm q^+\|_2 \braket{\bm D | \bm q^+}\epsilon$. From triangle inequality it is simple to check that $\norm{ \bm q^+ - \widetilde{\bm q}^+_{\rm est} }_2  \leq \norm{ \bm q^+}_2 \epsilon_1 + \widetilde{\norm{\bm q^+}}\epsilon_2 = \norm{ \bm q^+}_2 \epsilon_1 + {\norm{\bm q^+}}\epsilon_2 + \epsilon_1\epsilon_2$. 
With $\epsilon_1=\epsilon_2 \leq \frac{\braket{\bm D | \bm q}\epsilon}{2}$, we obtain that $ \| \bm q^+ - \widetilde{\bm q}^+_{\rm est} \|_2  \leq  \|\bm q^+\|_2 \braket{\bm D | \bm q^+}\epsilon$. The run time for obtaining $\ket{\widetilde{\bm q}^+_{\rm est}}$ is  
$\tOrd{K\frac{ \kappa(S) \Vert S\Vert_F}{\sqrt{\gamma}(\braket{\bm D | \bm q^+}\epsilon)^2} }$ and the run time for amplitude estimation is $\tOrd{ \frac{ \kappa(S)\Vert S\Vert_F }{\braket{\bm D | \bm q^+}\epsilon} }$.
\end{proof}
In conclusion, assuming a least-square market (Definition~\ref{def:assumeL2}) - which is weaker condition than the assumption of market completeness - there may exist special situations when there is the potential for an exponential speedup using the quantum linear systems algorithm. 

\section{Discussion and conclusions}
The study of quantum algorithms in finance is motivated by the large amount of computational resources deployed for solving problems in this domain.
This work explored the use of quantum computers for martingale asset pricing in one-period markets, with specific focus on incomplete markets. Contrary to other works in quantum finance, our algorithms are not directly based on quantum speedups of Monte Carlo subroutines~\cite{Montanaro2015}, albeit they use similar subroutines (e.g., amplitude amplification and estimation). One algorithm discussed here is based on the quantum zero-sum games algorithm (Theorem~\ref{thm:pricingzsg}), one is based on the pseudoinverse algorithm (Theorem~\ref{thm:pricingqla}), and one is based on the quantum simplex method (Theorem~\ref{thm:pricingsimplex}). They all output a price with relative error. The martingale measure obtained from the zero-sum game and the simplex algorithm returns a value that maximizes (or minimizes) the price over the convex set of martingale measures. Moreover, the quantum zero-sum game algorithm returns an $\epsilon$-feasible measure, which is to be interpreted as a solution where every constraint is satisfied up to a tolerance $\epsilon$. Theorem~\ref{thm:estimatingqqla} and Theorem~\ref{thm:pricingsimplex-pricingandextracting} allow to extract the martingale measure using the quantum simplex method and quantum pseudoinverse algorithm. We studied the conditions under which the quantum zero-sum game is better than the classical analogue and under which the simplex method is expected to beat the the zero-sum game method.  

In the quantum algorithms based on matrix inversion and the simplex method, there is a factor $1/ \braket{\bm D | \bm q}$ (for $\braket{\bm D | \bm q} \in (0, 1]$) in the run times. In many cases, this factor could be bounded by economic reasoning.  
For instance, put/call options that only with very small probability, say $\sim 1/K$, give a non-zero payoff, are usually not traded or modeled. Most common derivatives have a large enough number of events when they have some pay-out. Hence, the run time cost introduced by the factor $1/\braket{\bm D | \bm q}$ is not expected to grow polynomially with $N$ and $K$. We can draw similar conclusions for the factor $\rho = \frac{D_{\rm max}}{\Pi_D}$ for some derivatives. For example, a put option with a strike price of $150\$$ of a stock which trades at a market value of $100\$$, described by a BSM with $\mu=1$ and $\sigma=2$ has a price of $62.5\$$ dollars.  The value of $D_{\rm max}$ (using the discretizations used in the experiments of this work) is $148 \$$, making the value of $\rho \approx 2$. 

We remark again that all algorithms will not return an equivalent martingale measure per Definition~\ref{def:equivalencebetweenprobabilitymeasures}. The output of the zero-sum games algorithm will be $\tOrd{\frac{1}{\epsilon^2}}$ sparse, which follows directly from the total number of iterations. However, this algorithm is still good enough to return an approximation of the derivative price with the desired precision. The quantum algorithms based on the simplex and the pseudo-inverse algorithm are able to create quantum states that are $\epsilon$-close to an equivalent martingale measure.

In this work, we propose a new class of  market models called least-square markets (Definition~\ref{def:assumeL2}). This market condition, which is weaker than market completeness, is particularly interesting for quantum algorithms, as it may allow for large speedups in martingale asset pricing using the quantum linear systems algorithm. 
It remains to be seen if the market condition can be weakened further for a broader application of the quantum linear systems algorithm.
    
Along with the future research directions already discussed in the manuscript, we propose to study quantum algorithms for risk-neutral asset pricing in a parametric setting, i.e., when it is assumed that the martingale probability measure comes from a parameterized family of probability distributions. It is also interesting to estimate the quantum resources for the algorithms proposed in this manuscript in real-world settings (i.e., in a realistic market model for which the parameters $\bm \Pi$, $\bm \mu$, and $\bm \sigma$ have been estimated from historical data, realistic assumptions on a given hardware architectures like error correction codes, noise models, and connectivity, and so on).  We also leave as future work the study of other algorithms based on other quantum subroutines for solving linear programs, like interior-point methods \cite{kerenidis2018interiorpoint, casares2020quantum}.

\section{Acknowledgements}
We thank Nicolò Cangiotti for pointing us to useful citations in history of probability theory. 
Research at CQT is funded by the National Research Foundation, the Prime Minister’s Office, and the Ministry of Education, Singapore under the Research Centres of Excellence programme’s research grant R-710-000-012-135. We also acknowledge funding from the Quantum Engineering Programme (QEP 2.0) under grant NRF2021-QEP2-02-P05.

\printbibliography

\appendix

\input{appendix.tex}

\end{document}

%% file: appendix.tex
\section{Aspects of market modeling} \label{appMarket}
We briefly discuss three aspects of modeling markets, the sample space, the payoffs, and the no-arbitrage assumption.
First, an important question in probability theory is where the sample space comes from \cite{lijoi2011conversation,cox1946probability}. There is a ``true" underlying $\Omega_{\rm true}$, which describes the sample space affecting the financial assets. Each element would describe a state of the whole financial market or the outcome of a tournament. Especially in finance, this space is rarely observed or practical and one assumes certain models for the financial market. Such a model contains events which shall described by a sample space $\Omega_{\rm model}$, which we work with here. An example of this space are the trajectories of a binomial stochastic process. 

Second, a similar question as for the sample space arises regarding the availability of the payoffs $S$ (introduced below in Definition \ref{def:payoff-matrix}). There exists a ``true" mapping $S_{\rm true}$ which relates elements in $\Omega_{\rm true}$ to the financial asset prices. There will be some relation between the state of the financial market and the stock price of a company. 
As mentioned, usually one works with some model sample space $\Omega_{\rm model}$. Hence,
one considers a matrix $S_{\rm model}$ with relates events to price outcomes. Such a matrix can in some cases be generated algorithmically or from historical data. We discuss an example in Section \ref{secExampleBlackScholes} motivated by the Black-Scholes-Merton model. 

Third,  consider the notion of arbitrage which is the existence of investments that allow for winnings without the corresponding risk, see Definition \ref{def:arbitrage} in our context. The scope of the discussion on arbitrage depends if one considers the ``true" setting $(\Omega_{\rm true},S_{\rm true})$ or the model setting $(\Omega_{\rm model},S_{\rm model})$. For the ``true" market, one often assumes the efficient market hypothesis. While strongly debatable, the efficient market hypothesis states that all the information is available to all market participants and instantaneously affects the prices in the market. In other words, the market has arranged the current asset prices in a way that arbitrage does not exist. Of course, the existence of certain types of hedge funds and speculators proves that such an assumption is not realistic, especially on short time scales. However, it is often a reasonable first approximation for financial theory. In the model, the absence of arbitrage is a requirement rather than an assumption. If one assumes that the underlying true market is arbitrage-free (efficient-market hypothesis), then the model market should also be designed in an arbitrage-free way. If a model admits arbitrage opportunities, one needs to consider a better model. While the task of finding appropriate models for specific markets is outside the scope of this work, we show an example in Section \ref{secExampleBlackScholes}.

\section{Further notions and preliminary results}\label{app:appendices}

\subsection{Linear programming and optimization}\label{app:linprogoptimization}
We introduce the standard form of a LP. For an inequality vector $\bm x$ with $\bm x \geq \bm y$ we mean that $x_i \geq  y_i$ for all $i$. 
\begin{defn}[Standard form of linear program (LP) \cite{Boyd2004}]\label{def:lpstandardform}
For a matrix $A \in \mathbb{R}^{m \times n}$ where $m < n$, $c \in \mathbb{R}^m$, a linear program in standard form is defined as the following optimization problem:
  \begin{equation*}
    \begin{array}{ll@{}ll}
    \text{minimize}  & \bm{c}^{T} \bm x &\\
    \text{subject to}& A \bm x \leq \bm b   &  &\\
                     &  \bm x \geq 0.   & &
    \end{array}
    \end{equation*}
\end{defn}

Note that a minimization problem can be turned into a maximization problem simply by changing $\bm c$ to $-\bm c$.
It is often the case that an LP problem is formulated in a non-standard form, i.e. when we have 
inequality constraints different than $\leq 0$ or the equality constraint $A\bm x =\bm b$. It is usually simple to convert
a LP program in non-standard form into a LP in standard form. For instance, an equality constraint can be cast into two inequality constraint as $A\bm x=\bm b \Rightarrow A \bm x \leq \bm b$ and $-A\bm x \leq -\bm b$. 

Farkas' Lemma is a standard result in optimization \cite{ORIElectures,ORF523lectures}. We present here a version slightly adjusted to our needs. Specifically, we need the vector $\bm x$ to be strictly greater than $0$ in each component. This is needed because we need a probability measure $\bm q$ to be equivalent (as in Definition~\ref{def:equivalencebetweenprobabilitymeasures}) to the original probability measure (which is non-zero on the singletons of our $\sigma$-algebra). We state two theorems that we will use in the proof of Farkas' Lemma.  The first, is a lemma to the separating hyperplane theorem, while the second is about closedness and convexity of a cone. 

\begin{lemma}[{\cite[Example 2.20]{Boyd2004}}]\label{thm:separating}
Let $C \subseteq \mathbb{R}^{n}$ be a closed convex set and $\bm b \in \mathbb{R}^n$ a point not in $C$. Then $\bm b$ and $C$ can be strictly separated by a hyper plane, i.e. $\exists \bm y \in \mathbb{R}^n, \alpha \in \mathbb{R}, \text{ s.t. }$ $$\forall \bm z \in C, \: \bm{y}^{T}\bm z < \alpha, \bm y^{T}\bm{b} > \alpha.$$
\end{lemma}
The previous lemma states that a point and a closed convex set can be \emph{strictly} separated by a hyperplane: strict separation - i.e. having strict inequalities in both directions -  is not always possible for two closed convex sets \cite[Exercise 2.23]{Boyd2004}. Note that the direction of the inequalities can be changed by considering $-y$ instead of $y$. 

\begin{lemma}\label{lem:closedconvex}
Let $C:=\rm{cone}(\{  \bm{a}_1, \dots,   \bm{a}_n\}):=\{ \bm{s} \in \mathbb{R}^{m}: \bm s = \sum_{i=1}^K \lambda_i \bm{a}_i, \lambda_i \geq 0, \forall i \}$. Then $C$ is closed and convex.  
\end{lemma}

\begin{theorem}[Farkas' Lemma] Let $A \in \mathbb{R}^{m \times n}$, and $\bm b \in \mathbb{R}^{m}$. Then, only one of the following two proposition holds:
\begin{itemize}
\item $ \exists \: x \in \mathbb{R}^{n} : A\bm x=\bm b, \bm x > 0 $.
\item $ \exists \: y \in \mathbb{R}^{m} : A^T\bm y > 0,  \bm y^{T}\bm b < 0 $.
\end{itemize} 

\end{theorem}

\begin{proof}
We first show that both conditions cannot be true at the same time. We do this by assuming that both propositions holds true, and we derive a contradiction. By using the definition of $b$ in the second inequality, it is simple to obtain that $0 > b^Ty = x^TA^T y > 0$, which is a contradiction. The last inequality follows the positivity of an inner product between two vectors with strictly positive entries. 

Let us assume that the first case does not hold, then we prove that the second case must necessarily hold. We do this using  Corollary~\ref{thm:separating}. Let $ a_1, \dots,  a_K$ be the columns of $A$, and $C=\rm{cone}(\{  a_1, \dots,   a_K\})$ be the cone generated by the columns of $A$. By assumption, $b \not \in C$. For closedness and convexity of $C$, required by Corollary~\ref{thm:separating}, see Lemma~\ref{lem:closedconvex}. From the Corollary we have that $\exists\: y\in \mathbb{R}^{m}, y \neq 0$, and $\alpha \in \mathbb{R} $ such that $\forall z \in C, y^Tz > \alpha$ and $y^Tb < \alpha$.  Now we argue that $\alpha = 0$ by ruling out the possibility of it being positive or negative. Since the origin is part of the cone (i.e. $0 \in C$), we have that $\alpha$ cannot be bigger than $0$. But $\alpha$ cannot be negative either. Indeed, if there is a $x \in C$ for which $y^Tx <  0$ there must be another point $\lambda x$ in the cone (for $\lambda > 0$) for which $\lambda y^Tx < \alpha$, which contradicts Corollary~\ref{thm:separating}, and thus $\alpha = 0$. In conclusion, if $b \not \in C$, then for all the columns $a_i$ of $A$ holds that $y^Ta_i > 0$ and $y^Tb < 0$. \end{proof} 

We conclude this section with a standard result in linear algebra. 

\begin{theorem}[{\cite{strang1993introduction}}]\label{thm:ell2minimization}
Let $A\bm x=\bm b$ a linear system of equations for a matrix $A \in \mathbb{R}^{m \times n}$ with singular value decomposition $\sum_{i}^n\sigma_i \bm u_i \bm v_i^T=U\Sigma V^T$, where $U \in \mathbb{R}^{m \times n}$, $\Sigma \in \mathbb{R}^{n \times n}$ is a diagonal matrix with $\sigma_i \geq 0$ on the diagonal, and $V^T \in \mathbb{R}^{n \times n}$. The least-square solution of smallest $\ell_2$ norm exists, is unique, and is given by $$x^+:=A^{+}b := U\Sigma^+V^Tb.$$
Where $\Sigma^+$ is the square matrix having $1/\sigma_i$ on the diagonal if $\sigma_i > 0$, and $0$ otherwise.
\end{theorem}

\subsection{Random variables, measurable spaces}\label{app:probability}

\begin{defn}[$\sigma$-algebra \cite{tao2011introduction,kolmogorov1957elements}]\label{defn:sigma-algebra}
Let $\Omega$ be a set, and $\Sigma$ be a subset of the power set of $\Omega$ (or equivalently a collection of subsets of $X$). Then $(\Omega, \Sigma)$ is a $\sigma$-algebra if
\begin{itemize}
\item (Empty set) $\emptyset\in \Sigma$.
\item (Complement)  $\forall S \in \Sigma \Rightarrow \overline{S} := X\backslash S \in \Sigma$.
\item (Countable unions) $\Sigma$ is closed under countable union, i.e., if $S_1, S_2, \dots \in \Sigma$, then $\bigcup_{n=1}^\infty S_n \in \Sigma$.
\end{itemize}
\end{defn}

Observe that thanks to the de Morgan law, we can equivalently define the sigma algebra to be closed under countable intersection.
\begin{defn}
Let $\Omega$ be a set, and $\Sigma$ a $\sigma$-algebra. The tuple $(\Omega, \Sigma)$ is a measurable space (or Borel space). 
\end{defn}

\begin{defn}[Measurable function \cite{tao2011introduction}]
Let $(\Omega, \Sigma)$, and $(Y, T)$ two different measurable space. A function $f : \Omega \mapsto Y$ is said to be measurable if for every $E \in T$ 
$$f^{-1}(E):=\{x \in \Omega | f(x) \in E \} \in \Sigma.$$
\end{defn}
A measurable function is a function between the underlying sets of two measurable spaces that preserves the structure of the spaces: the pre-image of any measurable set is measurable. This is a very similar definition of a continuous function between topological spaces: the pre-image of any open set is open. 
\begin{defn}[Probability space]\label{defn:probability-space}
The tuple $(\Omega, \Sigma, \mathbb{P})$ is a probability space if
\begin{itemize}  
 \item $(\Omega, \Sigma)$ is a $\sigma$-algebra. 
 \item $\mathbb{P} : \Sigma \mapsto [0,1]$ is a measurable function such that
 \begin{itemize}
   \item $\mathbb{P}(\emptyset) =0 $ and $\mathbb{P}[S_i] \geq 0 \quad \forall S_i \in \Sigma$. 
   \item Let  $\{S_i\}_{i=1}^\infty $ be a set of disjoints elements of $\Sigma$. Then  $\mathbb{P}\left(\bigcup_{i=1}^\infty S_i\right) = \sum_{i=0}^\infty \mathbb{P}[S_i]$ is countably additive. 
     \item $\mathbb{P}(\Omega)=1$.
 \end{itemize}
\end{itemize}
\end{defn}
We say commonly that $\Omega$ is the set of \emph{outcomes} of the experiment, and $\F$ is the set of \emph{events}. That is, a probability space is a measure space whose measure of $\Omega$ is $1$. 

\begin{defn}[Equivalence between probability measures]\label{def:equivalencebetweenprobabilitymeasures}
  Let $(\Omega, \Sigma, \mathbb{P}), (\Omega, \Sigma, \mathbb{Q})$ two probability space with the same $\Omega$ and $\Sigma$. We say that $\mathbb{P}$ and $\mathbb{Q}$ are equivalent if and only if
  $$\forall A \in \Sigma \:\:\: \mathbb{P}(A)=0 \Leftrightarrow \mathbb{Q}(A)=0.$$
\end{defn}
Two equivalent measures agree on the possible and impossible events. Next, we recall the definition of and the formalism of random variable. 
\begin{defn}[Random variable]\label{def:random-variable}
A (real-valued) random variable on a probability space $(\Omega, \Sigma, \mathbb{P})$ is a measurable function $X: \Omega \mapsto \mathbb{R}$. 
\end{defn}
Usually the definition of a martingale is in the context of stochastic processes, but in the case of a discrete market in one time step, we can use a simplified definition.
\begin{defn}[Martingale]\label{def:martingale}
Let $X_i, i \in \{0,1\}$ be two random variables on the same probability space. They are a martingale if 
\begin{itemize}
    \item $\mathbb E[|X_i|] < \infty$.
    \item $\mathbb E[X_1 |X_0] = X_1$.
\end{itemize}
\end{defn}

\subsection{Quantum subroutines}\label{subsec:quantumsubroutines}
We need few results from previous literature in quantum linear algebra. 

\begin{theorem}[Pseudoinverse state preparation \cite{chakraborty2018power}]\label{cor:pseudoinv-vtaa}
Let $\kappa\geq 2$, and $H$ be an $N\times N$ Hermitian matrix such that the non-zero eigenvalues of $H$ lie in the range $[-1,-1/\kappa]\bigcup [1/\kappa,1]$. Suppose that for $\delta = \littleo{\eps/(\kappa^2\log^3\frac{\kappa}{\eps})}$ we have a unitary $U$ that is a $(\alpha,a,\delta)$-block-encoding of $H$ that can be implemented using $T_U$ elementary gates.  Also suppose that we can prepare a state $\ket{\psi}$ in time $T_{\psi}$ such that $\nrm{\Pi_{\mathrm{col}(H)}\ket{\psi}}\geq \sqrt{\gamma}$. Then there exists a quantum algorithm that outputs a state that is $\eps$-close to
$H^{+}\ket{\psi}/\nrm{H^{+}\ket{\psi}}$ at a cost 
$$\tOrd{\frac{\kappa}{\sqrt{\gamma}}\left(\alpha\big(T_U+a\big)\log^2\left(\dfrac{1}{\eps}\right)+ T_{\psi}\right)}.$$
\end{theorem}
The following theorem provides a relative estimate of the normalization factor from the previous theorem $\|H^+\ket{\psi}\|$. 
\begin{theorem}[Pseudoinverse state preparation and its norm estimation \cite{chakraborty2018power}]\label{cor:qls-vtae} 
Let $\eps>0$. Then under the same assumptions as in Corollary~\ref{cor:pseudoinv-vtaa}, there exists a variable time quantum algorithm that outputs a number $\Gamma$ such that 
$$1-\eps\leq \dfrac{\Gamma}{\nrm{H^{+}\ket{\psi}}}\leq 1+\eps,$$ 
at a cost
$$\tOrd{\dfrac{\kappa}{\eps\sqrt{\gamma}} \log\left(\dfrac{1}{\delta}\right) \left(\alpha\big(T_U+a\big)\log^2\left(\dfrac{1}{\eps}\right)+ T_\psi\right)},$$
with success probability at least $1-\delta$.
\end{theorem}

\begin{theorem} [$\ell_2$ state-vector tomography \cite{landman2019convolutional}] 
\label{thm:tomography}
    Given a unitary mapping $U_x: \ket{0} \mapsto \ket{x}$ in time $T(U_x)$ and $\delta>0$, there is an algorithm that produces an estimate $\overline{x} \in \mathbb{R}^m$ with $\norm{\overline{x}} = 1$ such that $\|x - \overline{x}\|_2 \leq \delta$ with probability at least $1 - 1/poly(m)$ using  $O\left(\frac{m\log m}{\delta^2}\right)$ samples from $\ket{x}$.
\end{theorem}

\begin{theorem}[Amplitude estimation \cite{BHMT00}]\label{thm:ampest}
Given a quantum algorithm $$A:\ket{0} \to \sqrt{p}\ket{y,1} + \sqrt{1-p}\ket{G,0},$$ where $\ket{G}$ is an arbitrary state and $p\in [0,1]$, then for any positive integer $t$, the amplitude estimation algorithm outputs $\tilde{p} \in [0,1]$ such that
$$
|\tilde{p}-p|\le 2\pi \frac{\sqrt{p(1-p)}}{t}+\left(\frac{\pi}{t}\right)^2,
$$
with probability at least $8/\pi^2$. It uses exactly $t$ iterations of the algorithm $A$. 
If $p=0$ then $\tilde{p}=0$ with certainty, and if $p=1$ and $t$ is even, then $\tilde{p}=1$ with certainty.
\end{theorem}

\begin{defn}[Quantum multi-vector access]\label{def:sampling_oracle}
Let $A \in \mathbb{R}^{m \times n}$.
We say that we have quantum multi-vector access to $A$ if we have access to a circuit that performs the mappings $\ket{i} \ket{0} \mapsto \ket{i}\ket{{A}_{i, \cdot}} = \ket{i}\frac{1}{\norm{{A}_{i,\cdot}}}\sum_{j=1}^n A_{ij}\ket{j}$, for all $i \in[m]$, and $\ket{0} \mapsto \frac{1}{\norm{A}_F}\sum_{i=1}^m \norm{{A}_{i,\cdot}} \ket{i}$. The cost to create the data structure needed to perform this mapping is $\Ord{mn\ {\rm poly} \log(mn)}$. If we update one entry of matrix $A$, the cost of updating the data structure is $\Ord{{\rm poly} \log(mn)}$. 
Building the data structure will also output $\|A\|_F$. 
The time needed for a query is $\Ord{{\rm poly} \log(mn)}$.
\end{defn}

\section{Conversion of primal pricing problem into the dual pricing problem}
\label{sec:conversionprimaldual}

We show how to convert the primal problem into the dual problem. We follow the standard approach of \cite{Boyd2004}. Recall that our problem has $\bm D \in \mathbbm R_+^K$, $S \in \mathbbm R_+^{N+1 \times K}$, and $ \bm{\Pi} \in \mathbbm R_+^{N+1}$. We assumed that there is an asset which is non-random. Thus, the first row of the matrix $S$ is composed of $1$s, and we use this to enforce a normalization constraint on $\bm q$. 
We work with the maximization LP in Eq.~(\ref{eq:simplexonthis}), which we express in a minimization form by multiplying the objective function by $-1$, as
\begin{equation}
  \begin{array}{ll@{}ll}
  \min\limits_{ \bm q \in \mathbbm R^K}  & -{\bm D}^T \bm q  &\\
  \text{s.t.}&  \bm q \geq 0 \\
    & S{\bm q} = {\Pi}. & &
\end{array}
\end{equation}
Let $\bm \lambda \geq 0$ be the KKT multipliers associated to the inequality constraints, and $\bm \xi$ the Lagrange multipliers associated to the equality constraints, the Lagrange function of this problem is
\be
\mathcal L( \bm q, \bm \lambda, \bm \xi) = - \bm{D}^T \bm{q} -\bm\lambda^T  \bm q + \xi^T(S \bm q -  \bm \Pi),
\ee
and the dual function is by definition
\be
g(\bm \lambda, \bm \xi) := \inf_{\bm q} \mathcal{L} \left(\bm q, \bm \lambda, \bm \xi \right)  = -\bm \xi^T\bm \Pi + \inf_{\bm q} (- \bm D^T - \bm \lambda^T + \xi^TS)\bm q.
\ee
The Lagrange dual problem of an LP in standard form is a maximization problem of the dual function, constrained on $\bm \lambda \geq 0$. 
\begin{equation}
  \begin{array}{ll@{}ll}
  \max\limits_{ \bm \xi \in \mathbbm R^{N+1}}  & g(\bm \lambda, \bm \xi)  &\\
  \text{s.t.}&  \bm \lambda \geq 0.
  \end{array}
  \end{equation}
By inspection of $g(\bm \lambda, \bm \xi)$, we can determine that it has value either $-\bm \xi^T\bm \Pi$ (when $ (- \bm D^T - \bm \lambda^T + \xi^TS) =0$) or $-\infty$, as in the latter case $\inf_{\bm q} (- \bm D^T - \bm \lambda^T + \xi^TS)\bm q$ is not bounded from below. Thus, the problem is non-trivial only in the first case, which we can obtain by embedding the condition of $- \bm D - \bm \lambda + S^T\xi = 0$ in the constraints of the linear program as
\begin{equation}
  \begin{array}{ll@{}ll}
  \max\limits_{\bm \xi \in \mathbbm R^{N+1}}  & - \bm \xi^T \bm \Pi  &\\
  \text{s.t.}&  \bm \lambda \geq 0 \\
  & S^T\bm \xi -\bm D - \bm \lambda = 0.
  \end{array}
\end{equation}
We can remove the inequality constraint and the sign from the maximization function, obtaining the minimization problem
\begin{equation}
  \begin{array}{ll@{}ll}
  \min\limits_{ \bm \xi \in \mathbbm R^N}  & \bm \xi^T \bm \Pi  &\\
  \text{s.t.}&   S^T\bm \xi -\bm D \geq 0.  \\
  \end{array}
\end{equation}
Hence, we have derived the dual given in Eq.~(\ref{eqDualMin}) of the pricing maximization problem given in Eq.~(\ref{eq:simplexonthis}).

\section{Numerical experiments of LP pricing in the BSM model}\label{app:BSMRadon}

In this section we present experiments of derivative pricing in the BSM model with a linear program \footnote{The code is available on \url{https://github.com/Scinawa/QLP-pricing}}. The linear program that we presented in Eq.~(\ref{eqRadonLP}) optimizes over the set of measure changes
\be
\mathcal M := \left \{ \bm x \in \mathbbm R^K \right \},
\ee
with the additional constraints given by the linear program.
On the other hand, the Black-Scholes formalism uses a particular measure from the one-parameter set of Gaussian measure changes (see Lemma~\ref{lem:grisanovgaussian})
\be \label{eqMBS}
\mathcal M_{\rm BS} := \left \{\bm x(\theta) \in \mathbbm R^K\ \big |\ x_\omega(\theta) = \exp\left (-\theta  B(\omega) - \frac{1}{2} \theta^2\right), \theta \in \mathbbm R, \omega \in[K] \right \}.
\ee
Optimizing over $\mathcal M$ finds a much larger range of prices, as the set $\mathcal M_{\rm BS}$ is of course much smaller than $\mathcal M$.  Via a regularization technique, we can find measures that are closer to the set $\mathcal M_{\rm BS}$. Specifically, we add a first derivative constraint. 
Note that 
for the Girsanov measure change it holds that 
\be \label{eqMReg}
\frac{x_{\omega+1} - x_{\omega-1}}{2 \Delta} \approx -\theta x_\omega.
\ee
This constraint allows to optimize over the set
\be
\mathcal M_{\rm reg}(\eta) := \left \{\bm x(\theta) \in \mathbbm R^K\ \big |\ \left\vert x_\omega -  x_{\omega+1}\right \vert \leq \eta x_\omega, \forall \omega \in \Omega \setminus \{K_0\} \right \},
\ee
where $\eta >0$. In the experiments as assume a payoff matrix of size $S\in \mathbb{R}^{2 \times 100}$, i.e. $N=1$ and $K=100$, with the first row representing risk-less asset with payoff $1$ in all the events. In our plots, we start with fixing the parameters $\Pi=10$, drift $\mu=1$ (and $\mu=0$ for Figure~\ref{figB:muandsigma}), volatility $\sigma=1$ and strike price for the derivative of $Z=10$.  In most of the plots, the analytic value of the derivative is always within the range of the admissible prices of the minimization and maximization problem. In Figure~\ref{figA:regularization} we test different value for the regularization parameter. In picture (a) we find that the regularization parameter in the range of $\eta \in \{0.2, 0.5, 1\}$ allows to get close to the analytic solution of the Radon-Nikodym derivative. In figure (b) we scan the value of the regularization parameter for $\eta \in [0.001, 10]$. We observe that in our example, changing $\eta$ makes the range of admissible prices within $\pm 1\%$ of the analytic price. Next, we turn our attention to Figure~\ref{figB:muandsigma}, i.e. the changes in the drift ($\mu$) and volatility ($\sigma$)  of Eq.~(\ref{eqBSMdiscrete}) in the main text.  We observe that changing the drift does not influence the analytic price of the derivative. This is because the change of measure makes the price of the call option independent of $\mu$. After the value of $\mu$ for which the price of the minimization and maximization program coincide, the two LPs start to have no feasible solutions. Scanning the volatility parameter does not influence the range of admissible prices, no matter which regularization constraint we use. In the last plots of Figure~\ref{figC:strikestock} we plot the value of the derivative while scanning the value of the strike price (in figure (a)) and the stock price (in figure (b)). The analytic price is well approximated by the minimum and the maximum price obtained by the solution of the two LP programs. 

Concluding, it is possible to obtain estimates for the range of admissible prices using the linear programs for most of the ranges of parameters. 
\begin{figure}
\centering
\begin{subfigure}[b]{0.49\textwidth}
\centering
\includegraphics[width=\textwidth]{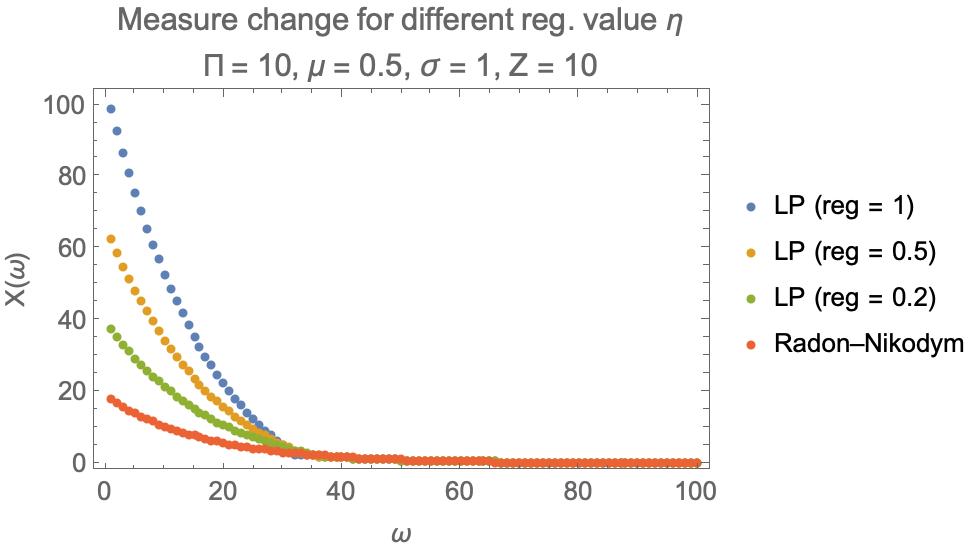}
\label{figA1}
\end{subfigure}
\begin{subfigure}[b]{0.49\textwidth}
\centering
\includegraphics[width=\textwidth]{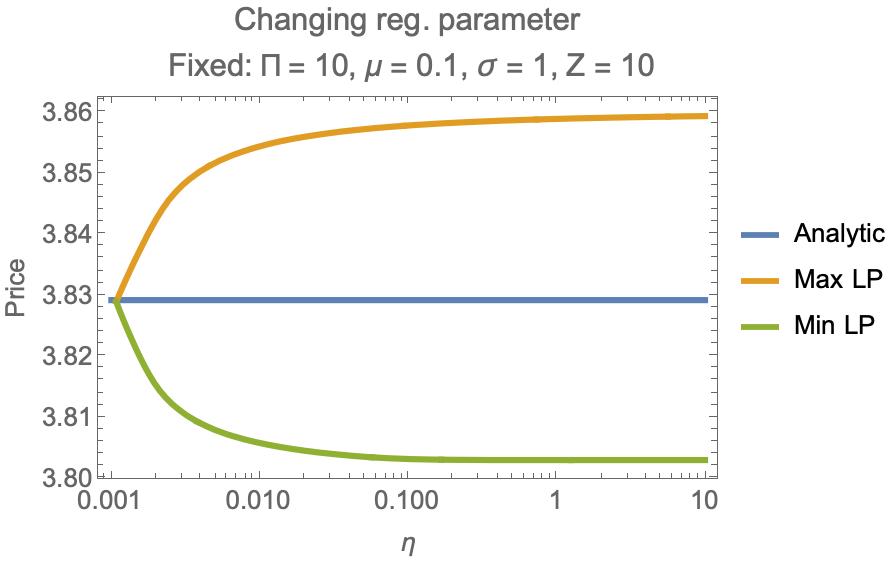}
\label{figA2}
\end{subfigure}
\caption{Changing regularization parameters. 
(Left panel) The Radon-Nikodym derivative obtained analytically is compared to the solution vector obtained from solving the linear program.
(Right panel) The regularization parameter constrains the solution space and hence limits the admissible prices of the derivative.}
\label{figA:regularization}
\end{figure}

\begin{figure}
\centering
\begin{subfigure}[b]{0.44\textwidth}
\centering
\includegraphics[width=\textwidth]{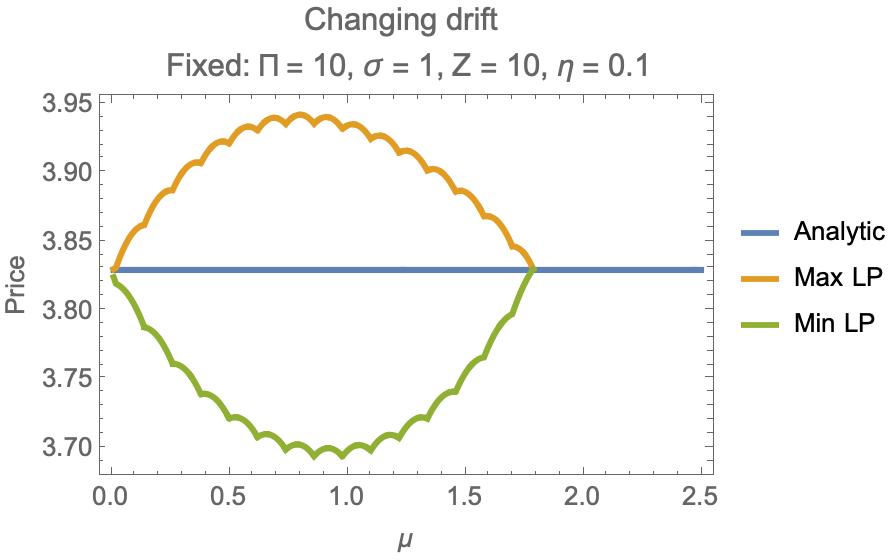}
\label{figB1}
\end{subfigure}
\hfill
\begin{subfigure}[b]{0.44\textwidth}
\centering
\includegraphics[width=\textwidth]{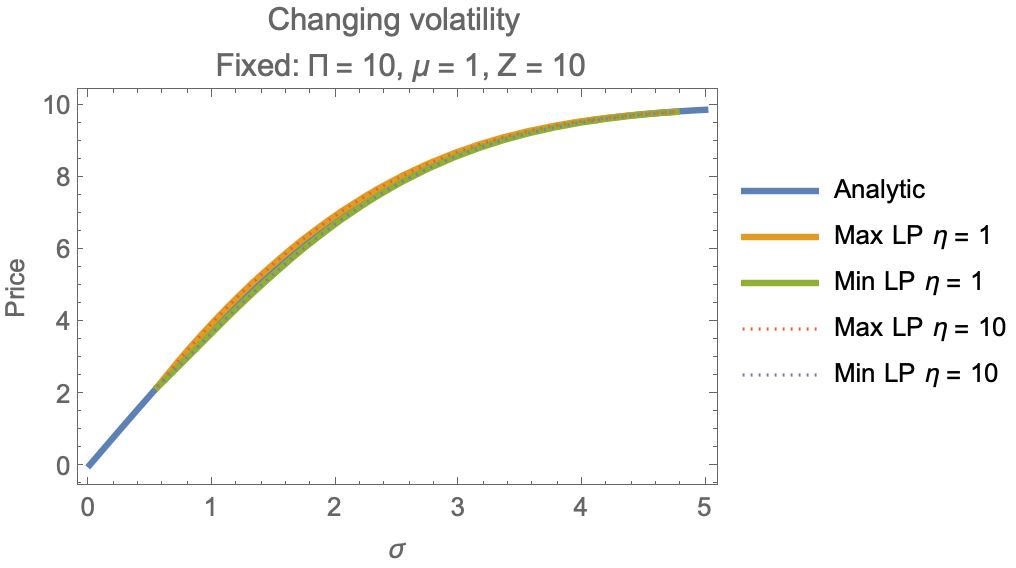}   
\label{figB2}
\end{subfigure}
\caption{Range of admissible prices changing $\mu$ and $\sigma$. (Left panel) The change in admissible prices while changing the drift parameters $\mu$. Note that the drift does not change the analytical value of the derivative.
(Right panel) The change in admissible prices while changing the volatility parameters $\sigma$.}
\label{figB:muandsigma}
\end{figure}

\begin{figure}
\centering
\begin{subfigure}[b]{0.44\textwidth}
\centering
\includegraphics[width=\textwidth]{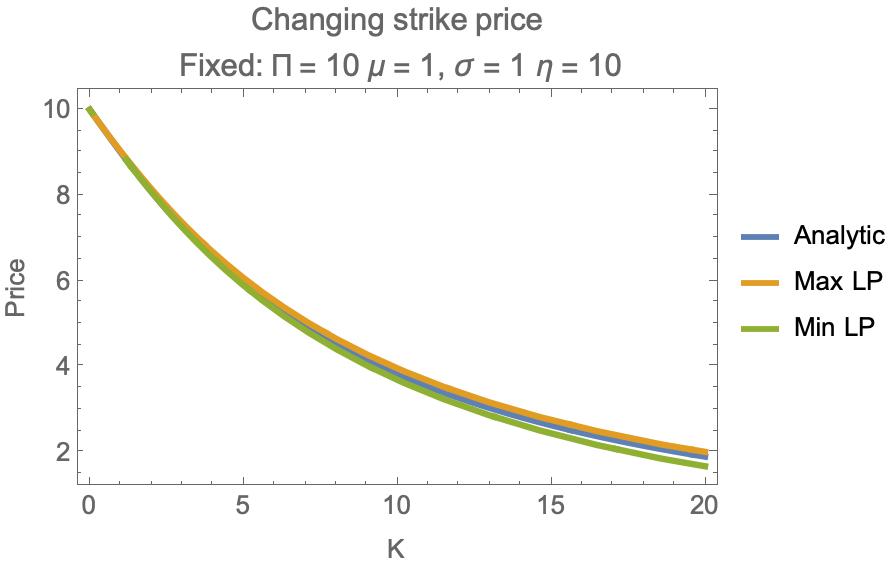}
\label{fig1:y equals x}
\end{subfigure}
\hfill
\begin{subfigure}[b]{0.44\textwidth}
\centering
\includegraphics[width=\textwidth]{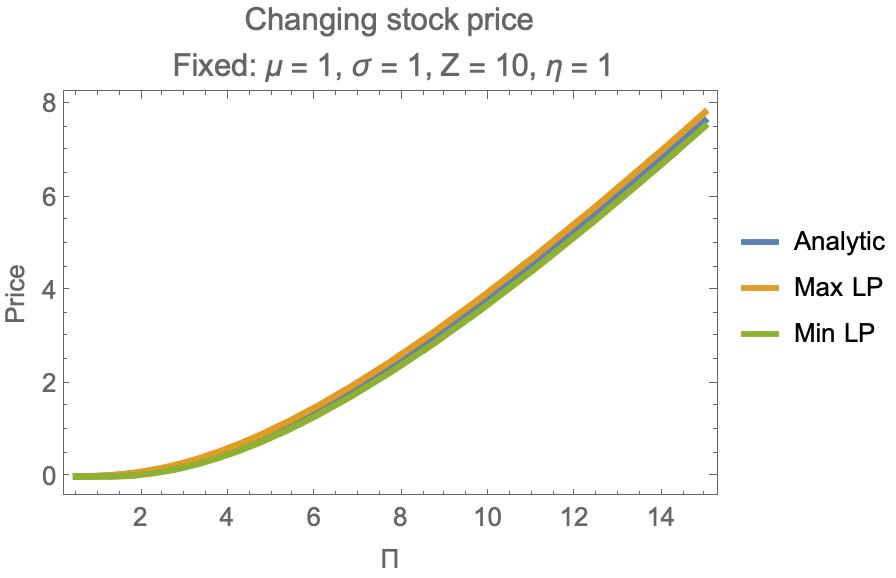}

\label{fig2:three sin x}
\end{subfigure}
\caption{
In the left panel, changes of price of the derivative while changing the strike price. In the right panel, 
changes of value of the derivative while changing the stock price.
}
\label{figC:strikestock}
\end{figure}

\section{Quantum simplex algorithm for risk-neutral pricing}\label{sec:simplexpricing}
The simplex algorithm is often the algorithm of choice for solving linear programs in many practical applications. A quantum analog of the simplex algorithm has recently been proposed \cite{nannicini2019fast,nannicini2021fastPUBLISHED}, and we can study its runtime on the linear program in Eq.~(\ref{eq:simplexonthis}). 

We recall that the standard assumption of the simplex method is that the number of constraints is smaller than the number of variables, i.e. the matrix describing the linear program is full row rank. This restricts the applicability of the simplex algorithm to incomplete markets. The simplex algorithm has exponential run time in the worst case. While it is well-known that the theoretical worst-case number of iterations grows exponentially in the number of variables, $\Ord{2^{K}}$ in our case, in practical instances the run time is much more favorable. 
The quantum simplex algorithm is expected to inherit the good properties of the classical algorithm. 

A feature of this quantum algorithm is that all the subroutines of the algorithm output just a single number. Framing the algorithm in this way, it is possible to skip expensive tomography operations that would come from a straightforward quantization of the simplex algorithm.
A rigorous but succinct description of the quantum procedures for the simplex algorithm can be found below, which  describes a formalization of the work of \cite{nannicini2019fast, nannicini2021fastPUBLISHED}. The quantum simplex algorithm can operate with two different quantum access to the linear program. The first is the quantum multi-vector access (Definition~\ref{def:sampling_oracle}), and the second is the adjacency array model (Definition~\ref{oracle:adarray}).

In the case where the matrix of the LP and the vectors can be accessed through quantum multi-vector access, the run time of these subroutines can be substantially improved compared to the case when we have only query access to the vectors and matrices describing the problem. Based on \cite{nannicini2019fast, nannicini2021fastPUBLISHED}, we report here the run time of the quantum simplex algorithm for our problem. In the language of the simplex algorithm, a basis is a set of linearly independent columns of $S$ from which we can obtain the solution to our optimization problem, as we show in Theorem~\ref{thm:pricingsimplex-pricingandextracting}.

 We need to introduce more notation to understand Algorithm \ref{algNANNI} and Theorem \ref{thm:nannicini}. We define $d_c(A)$ the maximum number of nonzero elements of the column of a matrix $A$, $d_r(A)$ the maximum number of nonzero elements of the rows of $A$, and $d(A)$ the maximum between the two. Consider a matrix $A \in \mathbb{R}^{n \times m}$. As standard in simplex literature, we denote with $B \subset \{1, \dots, n \}$ a set of size $m$ of linearly independent column of $A$.  With $A_B$ we denote the square invertible matrix composed by set of $m$ linearly independent column of $A$ indexed by the set $B$. Likewise, $A_N$ is the remaining set of column of $A$ that are not in $B$. With $T_{LS}(M, b, \epsilon)$ we denote the runtime of the subroutine for solving linear system of equation on a quantum computer. The runtime of this subroutine is dependent on the access model on the matrix $M$ and the vector $b$. In the rest of this section we better formalize the result of \cite{nannicini2019fast, nannicini2021fastPUBLISHED} so to ease the comparison with the zero-sum games approach. Recall that the fastest classical algorithm for a single iteration of the simplex (assuming we use the matrix multiplication algorithm of \cite{yuster2005fast}) has a cost of $O(d_c^{0.7}m^{1.9} + m^{2+o(1)} + d_cn)$ per iteration (which is dominated by the so-called pricing step of the classical algorithm, which we define below). In \cite{nannicini2019fast}, there are also quantum subroutines for implementing some heuristics (like Steepest Edge rule for column generation) that are usually employed in classical simplex method, which we do not discuss here. As for the classical case, we are sometimes willing to spend more time for getting a better suggestion for the column that can enter the basis, if we can decrease the number of iterations of the whole algorithm. We report definition of the adjacency array model below, while the definition of the multi-vector access is reported in Definition~\ref{def:sampling_oracle} in the Appendix. The adjacency array model is used to access sparse matrices. While the payoff matrices obtained from portfolios of stocks might are not expected to be sparse, payoff matrices obtained from portfolios of certain class of derivatives might me sparse.

 \begin{defn}[Adjacency array model \cite{Childs2017linear}] \label{oracle:adarray}
We say that we have quantum access in the adjacency array model for a matrix $V \in \mathbb{R}^{n \times d}$ if there is an oracle that allows to perform the mappings: 
\begin{itemize}
\item   $\ket{i}\mapsto\ket{i}\ket{s(i)}$ where $s(i)$ is the number of entries in row $i$, 
\item $\ket{i,l}\mapsto\ket{i, \nu(i,l)}$, where $\nu(i,l)$ is the $l$-th nonzero entry of the $i$-th row of $V$, and
\item $\ket{i,j}\ket{i,j,V_{ij}}$ as in the standard query model. 
\end{itemize}
\end{defn}

\begin{algorithm}[H]
  \caption{Quantum simplex iteration}
    \label{algNANNI}
  \begin{algorithmic}[1]
    \Require{Query access to matrix $A$, basis $B$, cost vector $c$, precision parameters $\epsilon, \delta, t$.}
\Ensure{A flag ``optimal'', a flag ``unbounded'', or a pair (k, $\ell$) where $k$ is a nonbasic variable with negative reduced cost, $\ell$ is the basic variable that should leave the basis if $k$ enters}. 
\State Normalize vector of reduced cost $c$ so that $\|c_B\|=1$. Normalize $A$ so that $\|A_B\|\leq 1$.  
\State Apply $IsOptimal(A,B,\epsilon)$ to determine if the current basis is optimal. If so, return ``optimal''
\State Apply $FindColumn(A,B,c,\epsilon)$ to determine a column with negative reduced cost. Let $k$ be the column index returned by the algorithm. Then $k$ is the pivot column. 
    \State Apply $IsUnbounded(A_B, A_k, \delta)$ to determine if the problem is unbounded. If so, return "unbounded".
    \State Apply $FindRow(A_B, A_k, b, \delta, t)$ to determine the index $\ell$ of the row that minimies the ratio test. $\ell$ is the pivot row. 
    \State Update the basis $B \leftarrow (B \symbol{92} \{B(\ell)\} \cup \{k\}  )$
  \end{algorithmic}
\end{algorithm}

This algorithm has improved runtimes, in the stronger access model described in Definition~\ref{def:sampling_oracle}. 

\begin{theorem}[Fast quantum subroutine for simplex \cite{nannicini2019fast, nannicini2021fastPUBLISHED}]\label{thm:nannicini}
Let $\epsilon, \delta, t > 0$ be precision parameters.
Let $A \in \mathbb{R}^{m \times n}$ the constraint matrix, $c \in \mathbb{R}^{n}$ the cost vector, $b \in  \mathbb{R}^{m}$ be an LP in the following form
\begin{equation}
  \begin{array}{ll@{}ll}
  \max\limits_{\bm x}  & \bm{c}^T \bm{x}  &\\
  \text{s.t.}& A\bm{x} = \bm{b}\\
    & \bm x \geq 0. & &
\end{array}
  \end{equation}

Let $A_B \in \mathbb{R}^{m \times m}$ be a square invertible matrix obtained from a base $B \subseteq [n]$ of columns of $A$, and $A_k$ any of the columns of $A$. The quantum subroutines in Algorithm~\ref{algNANNI} can:
\begin{itemize}
  \item find a column with negative reduced cost ($FindColumn$), and identify if a basis is optimal ($IsOptimal$) has gate complexity $\tOrd{\frac{\sqrt{n}}{\epsilon}T_{LS}(A_B, A_k, \frac{\epsilon}{2}) + m}$ where $A_k$ is any column of $A$: \cite[Proposition 1, Proof of Theorem 5, Section 5.3]{nannicini2019fast}
  \begin{itemize} 
  \item with Definition~\ref{oracle:adarray}, the time complexity is 
  $\tOrd{\frac{1}{\epsilon}\kappa(A_B)d(A_B)\sqrt{n}(d_c(A) n + d(A_B)m)}$,
  \item with Definition~\ref{def:sampling_oracle}, the time complexity is $\tOrd{\frac{\sqrt{n}}{\epsilon}(\kappa(A_B)\mu(A_B)+m)}$    \end{itemize}
\item identify if a nonbasic column proves unboundedness ($IsUnbounded$) in gate complexity $\tOrd{\frac{\sqrt{m}}{\delta}T_{LS}(A_B, A_k, \delta)}$: \cite[Proposition 8, Theorem 2]{nannicini2019fast}
\begin{itemize} \item with Definition~\ref{oracle:adarray}, the time complexity is $\tOrd{\frac{\sqrt{m}}{\delta}(\kappa(A)d^2(A)m)}$, \item  with Definition
\ref{def:sampling_oracle}, the time complexity is $\tOrd{\frac{t}{\delta}\kappa(A)m}$;
\end{itemize}
\item perform the ratio test ($FindRow$) in gate complexity $\tOrd{\frac{t\sqrt{m}}{\delta}\left(T_{LS}(A_B, A_B, \frac{\delta}{16}) + T_{LS}(A_B, b, \frac{\delta}{16})\right)}$: \cite[Section 6, Theorem 2]{nannicini2019fast}.
\begin{itemize}
    \item with Definition~\ref{oracle:adarray} the time complexity of this step is $\tOrd{\frac{t}{\delta}k(A_B)d^2m^{1.5}}$, 
    \item with Definition~\ref{def:sampling_oracle}, the time complexity is $\tOrd{\frac{t\sqrt{m}}{\delta}\kappa(A_B)\mu(A_B) }$.
    
\end{itemize}
\end{itemize}
In the previous statements, $t$ is a relative error for the precision parameter determining how well in $FindRow$ we approximate the minimum of the ratio test performed by a classical algorithm, $\epsilon$ is the optimality tolerance for the computation of the reduced cost (i.e. $\overline{c}_N^T  = c_N^T - c_B^TA_B^{-1}A_N \geq 0 \geq  -\epsilon$) computed in a single iteration of the algorithm, and $\delta$ is the precision used to determine if the problem is unbounded, (i.e. for a given $k\in[n]$, check if $A_B^{-1}A_k \leq 0 \leq \delta$). In case we use Definition~\ref{def:sampling_oracle}, we have $\widetilde{O}(m)$ classical operations to update the data structure. The normalization steps, and the update to the data structure performed at the beginning of each iteration are dominated by the runtime of the quantum subroutines. The cost of a single iteration is dominated by the pricing step (i.e. $FindColumn$ subroutine). 
\end{theorem}

\begin{theorem}[Quantum simplex algorithm for martingale pricing]\label{thm:pricingsimplex}
Let $(\bm \Pi, S)$ be a price system, for a matrix $S \in \mathbb{R}^{N \times K}$, and $\bm \Pi \in \mathbb{R}^{N}$, and let $D$ be a derivative. 
Let $\epsilon_{\rm{smpx}} \in (0, 1)$ a problem-dependent precision parameter. There is an iterative quantum algorithm that finds a matrix $S_B \in \mathbbm R^{N\times N}$, with quantum access to $S$ and $\bm \Pi$ in the same number of iterations of the classical simplex algorithm. 
  \begin{itemize} 
  \item Assuming the adjacency array model (Definition~\ref{oracle:adarray}), an iteration of the algorithm has gate complexity of 
  
  $\tOrd{\frac{1}{\epsilon_{\rm{smpx}}}\kappa(S_B)d(S_B)\sqrt{K}\left(d_c(S) K + d(S_B)N\right)}$, and $O(K)$ classical operations.
  \item Assuming quantum multi-vector access (Definition~\ref{def:sampling_oracle}), the gate complexity is $\tOrd{\frac{\sqrt{K}}{\epsilon_{\rm{smpx}}}\left(\kappa(S_B)\|S_B\|_F+N\right)}$ per iteration, and $O(Nd(S)+K)$ classical operations, 
  \end{itemize}
where $S_B$ is an optimal basis formed with linearly independent columns of $S$ chosen during the execution of the algorithm, $d(S)$ is the sparsity of the matrix $S$, and $d_c(S)$ is the maximum number of nonzero entries in any column of $S$. 
\end{theorem}

Once a basis is selected, we can obtain the martingale measure as $\bm q = S_B^{-1}\bm \Pi$. Thus, we can use a quantum algorithm to recover an estimate of $\bm q$, that we can use for pricing other derivatives. 

\begin{theorem}\label{thm:pricingsimplex-pricingandextracting}
In the same setting of Theorem~\ref{thm:pricingsimplex}, assume you have found a basis $S_B$, and assume to have quantum multi-vector access to it. For $\epsilon > 0$, there is a quantum algorithm that with high probability estimates the price of a derivative $D$ with relative error $\epsilon$ in  $\tOrd{\frac{\kappa(S_B)\|S_B\|_F}{\epsilon}\frac{\| \bm D\|_2}{\Pi_{D, \bm q^+} } }$ queries to the oracle, and returns an estimate of $\bm q$ in time $\tOrd {K\frac{ \kappa(S)\|S_B\|_F}{\epsilon^2}\frac{\| \bm D\|^2_2}{\Pi^2_{D, \bm q^+} } }$.
\end{theorem}

We skip the proof of this theorem, as it is similar to the proof of Theorem~\ref{thm:pricingqla} and Theorem~\ref{thm:estimatingqqla}. Note that (contrary to the Theorem \ref{thm:pricingqla} and \ref{thm:estimatingqqla}) in Theorem~\ref{thm:pricingsimplex-pricingandextracting} there is no factor $\sqrt{\gamma}$ at the denominator of the run time, as the matrix $S_B$ is invertible by definition. A few comments on the run time and the comparison with the quantum zero-sum games algorithm. As we work under the assumption that $\|S_B\| \leq 1$, $\|S_B\|_F\leq \sqrt{N}$, and the run time with Definition~\ref{def:sampling_oracle} is $O(\frac{\kappa(S_B)}{\epsilon_{\rm{smpx}}}\sqrt{NK}+N)$ per iteration. In our setting it seems to be hard to give theoretical arguments to bound $\kappa(S_B)$ and $\|S\|_F$. Without market completeness, i.e. when $S$ has more columns than rows, there are no trivial bounds on $\kappa(S_B)$ from $\kappa(S)$, as one can create an ill-conditioned matrix by arbitrarily picking linearly dependent columns from $S$. In this case, instead of assuming the (potentially infinite) run time given by the unbounded condition number, one could work with a quantum algorithm that performs linear algebraic operations by discarding the singular values smaller than a certain threshold. This will insert further approximation errors in the solution of the linear system solved with the quantum computer. That this further approximation error is good enough for the application to the simplex method is a subject that requires further exploration, and we leave this for future work. 

\paragraph{Comparison between quantum simplex and quantum zero-sum game algorithm.} We conclude this section by observing that the quantum simplex method becomes convenient compared to quantum algorithms for zero-sum games when
\begin{equation}
  ( \sqrt{K} + \sqrt{N} )\left( \frac{(r+1)\Pi_D}{\epsilon_{\rm zsg}}  \right)^3 
  \geq \frac{\kappa(S_B) \sqrt{KN}}{\epsilon_{\rm smpx}} T_{\rm smpx},
\end{equation}
where $T_{\rm smpx}$ is the number of iterations of the simplex algorithm, and $\epsilon_{\rm zsg}$ is the precision required in the computation of the solution of the zero-sum game algorithm. The parameter of the simplex algorithm $\epsilon_{\rm smpx}$ is an internal precision parameter used within the subroutines and is expected to be constant with respect to the dimensionality of the problem. A possible heuristic to estimate if one algorithm might outperform the other is to check if:

\begin{equation}
     \frac{\sqrt{NK}}{\sqrt{N}+\sqrt{K}} \leq 
     \left( \frac{\epsilon_{\rm smpx}}{\kappa(S_B)T_{\rm smpx}} \right) 
    \left(\frac{(r+1)\Pi_D}{\epsilon_{\rm zsg}} \right)^3.
\end{equation}